\documentclass[10pt, reprint, aps, pra]{revtex4-1}
\usepackage[utf8]{inputenc}
\usepackage[english]{babel}
\usepackage{amsmath}
\usepackage{amsfonts}
\usepackage{amssymb}
\usepackage{makeidx}
\usepackage{graphicx}
\usepackage{amsthm}
\usepackage{color}
\newtheorem{prop}{Proposition}
\newtheorem{coro}{Corollary}
\theoremstyle{remark}
\newtheorem{exm}{Example}

\theoremstyle{definition}
\newtheorem{defin}{Definition}
\def\states{\mathfrak{S}}
\def\dens{\mathfrak{D}}
\def\chan{\mathfrak{C}}

\def\Ha{\mathcal{H}}

\def\conv{\mathit{conv}}

\def\aff{\mathit{aff}}
\def\ri{\mathit{ri}}

\def\Tr{\text{Tr}}
\usepackage{dsfont}
\def\I{\mathds{1}}
\def\tmin{\dot{\otimes}}
\def\tmax{\hat{\otimes}}
\def\treal{\tilde{\otimes}}
\def\<{\langle}
\def\>{\rangle}
\def\dH{\dim(\mathcal{H})}
\begin{document}
%
%
\title{Conditions for the compatibility of channels in general probabilistic theory and their connection to steering and Bell nonlocality}

\author{Martin Plávala}
\email{martin.plavala@mat.savba.sk}
\affiliation{Mathematical Institute, Slovak Academy of Sciences, \v Stef\' anikova 49, Bratislava, Slovakia}

\begin{abstract}
We derive general conditions for the compatibility of channels in general probabilistic theory. We introduce formalism that allows us to easily formulate steering by channels and Bell nonlocality of channels as generalizations of the well-known concepts of steering by measurements and Bell nonlocality of measurements. The generalization does not follow the standard line of thinking stemming from the Einstein-Podolsky-Rosen paradox, but introduces steering and Bell nonlocality as entanglement-assisted incompatibility tests. We show that all of the proposed definitions are, in the special case of measurements, the same as the standard definitions, but not all of the known results for measurements generalize to channels. For example, we show that for quantum channels, steering is not a necessary condition for Bell nonlocality. We further investigate the introduced conditions and concepts in the special case of quantum theory and we provide many examples to demonstrate these concepts and their implications.
\end{abstract}

\maketitle

\section{Introduction}
Incompatibility of measurements is the well-known quantum phenomenon that gives rise to steering and Bell nonlocality. Historically, the idea of measurement incompatibility dates back to Bohr's principle of complementarity. Steering was first described by Schr\"{o}dinger \cite{Schrodinger-steering} and Bell nonlocality was first introduced by Bell \cite{Bell-ineq}, both as a reply to the paradox of Einstein, Podolsky and Rosen \cite{EinsteinPodolskyRosen-paradox}. It is known that incompatibility of measurements is necessary and in some cases sufficient for both steering and Bell nonlocality, but the operational connection between incompatibility, steering and Bell nonlocality was so far not described in general terms that would also fit channels, not only measurements.

There was extensive research into properties of quantum incompatibility of measurements \cite{HeinosaariReitznerStano-compatibility, Hsu-compatibility}, quantum incompatibility of measurements and its noise robustness, or degree of compatibility \cite{HeinosaariKiukasReitzner-compatNoise, HeinosaariSchultzToigoZiman-maxInc}, connection of quantum incompatibility of measurements and steering \cite{BavarescoQuintinoGueriniMacielCavalcantiCuhna-compatSteering, CavalcantiSkrzypczyk-steering, UolaMoroderGuhne-steering, UolaBudroniGuhnePellonpa-steeringCompatibility, QuintinoVertresiBrunner-steering}, connection of quantum incompatibility of measurements and Bell nonlocality \cite{WolfPerezgarciaFernandez-measIncomp, BanikGaziGhoshKar-maxCHSH, HirschQuintinoBrunner-noBell, ErikaVertesi-noBell} and connection between steering and Bell nonlocality \cite{GirdharCavalcanti-steeringBell, CavalcantiSkrypczyk-compatibilitySteeringBell}, for a recent review see \cite{HeinosaariMiyaderaZiman-compatibility}. In recent years, the problems of incompatibility of measurements on channels \cite{SedlakReitznerChiribellaZiman-compatibility}, compatibility of channels \cite{HeinosaariMiyadera-compOfChan}, the connection of channel steering to measurement incompatibility \cite{BanikDasMajumdar-channelSteering} and incompatibility in general probabilistic theory \cite{FilippovHeinosaariLeppajarvi-GPTcompatibility, Banik-steering, StevensBusch-maxCHSH, Jencova-compatibility} were all studied.

The aim of this paper is to heavily generalize the recent results of \cite{Jencova-compatibility}, where compatibility, steering and Bell nonlocality of measurements were formulated using convex analysis and the geometry of tensor products. In this paper, we will generalize the ideas and results of \cite{Jencova-compatibility} for the case of two channels in general probabilistic theory. The generalizations are not straightforward and we will have to introduce several new operational ideas and definitions, e.g. we introduce the operational interpretation of direct products of state spaces and we define steering and Bell nonlocality as very simple entanglement-assisted incompatibility test, that boil down to the problem whether there exists a multipartite state with given marginal states.

During all of our calculations we will restrict ourselves to finite-dimensional general probabilistic theory and to only the case of two channels. We will restrict to only two channel just for simplicity, as one may easily formulate many of our results for more than two channels using the same operational ideas as we will present.

The paper is organized as follows: in Sec. \ref{sec:motivation} we describe our motivation for using general probabilistic theory. We provide several references to known applications and their connections to each other. In Sec. \ref{sec:GPT} we introduce general probabilistic theory. Note that in subsection \ref{subsec:GPT-direct} we introduce the operational interpretations of direct products in general probabilistic theory. In Sec. \ref{sec:compat} we define compatibility of channels and we derive a condition for compatibility of channels. In Sec. \ref{sec:meascompat} we show that our condition for compatibility of channels yields the condition for compatibility of measurements that was presented in \cite{Jencova-compatibility}. In Sec. \ref{sec:quantcompat} we derive specific conditions for the compatibility of quantum channels. In Sec. \ref{sec:prelude} we propose an idea for a test of incompatibility of channels, that will not work at first, but will eventually lead to both steering and Bell nonlocality. In Sec. \ref{sec:steering} we define steering by channels as one-side entanglement assisted incompatibility test and we derive some basic results. In Sec. \ref{sec:meassteering} we show that for the special case of measurements our definition of steering leads to the  standard definition of steering \cite{WisemanJonesDoherty-nonlocal} in the formalism of \cite{Jencova-compatibility}. In Sec. \ref{sec:quantsteering} we derive the specific conditions for steering by quantum channels, we show that every pair of incompatible channels may be used for steering of maximally entangled state and that there are entangled states that are not steerable by any pair of channels, among other results. In Sec. \ref{sec:nonlocal} we define Bell nonlocality of channels as a two-sided entanglement assisted incompatibility test and we derive some basic results, then in Sec. \ref{sec:measnonlocal} we show that, when applied to measurement, the general definition of Bell nonlocality yields the standard definition of Bell nonlocality \cite{WisemanJonesDoherty-nonlocal} in the formalism of \cite{Jencova-compatibility} and we also show that for measurements steering is a necessary condition for Bell nonlocality. In Sec. \ref{sec:quantnonlocal} we derive conditions for the Bell nonlocality of quantum channels, we formulate a generalized version of the CHSH inequality, we show that for such inequality Tsirelson bound \cite{Cirelson-bound} both holds and is reached, we show an example of violation of the generalized version of CHSH inequality and we build on the example from Sec. \ref{sec:quantsteering} of an entangled state not steerable by any pair of channels to show that, even though the state is not steerable by any pair of channels, it leads to Bell nonlocality, which shows that steering is not a necessary condition for Bell nonlocality for quantum channels. In Sec. \ref{sec:conclusions} we conclude the paper by presenting the many open questions and possible areas of research opened by our paper.

\section{Motivations for using general probabilistic theory} \label{sec:motivation}
There are few motivations to using general probabilistic theory. The first motivation is mathematical as general probabilistic theory is a unified framework capable of describing both classical and quantum theory, as well as other theories. In the current manuscript the mathematical motivation is (according to the personal opinion of the author) even stronger as some of the formulations of the presented ideas and some of the proofs of the presented theorems turn out to be clearer in the framework of general probabilistic theory.

The second motivation comes from foundations of quantum theory as general probabilistic theory provides insight into the structure of entanglement and incompatibility.

The third and most promising motivation comes from information theory. There were developed several models \cite{Barrett-infProc, Spekkens-toyTheory, PopescuRohrlich-PRbox} that have very interesting information-theoretic properties and that can be described by general probabilistic theory, albeit sometimes it needs to be extended even more \cite{JanottaLal-noRestriction}. Apart from the well-known results on the properties of Popescu-Rohrlich boxes \cite{vanDam-complexity, QuekShor-communication}, it was showed that there are theories in which one can search a $N$-item database in $O(\sqrt[3]{N})$ queries \cite{Aaronson-localVariables} and that there is a general probabilistic theory that can be simulated by a probabilistic classical computer that can perform Deutsch-Jozsa and Simon's algorithm \cite{JohanssonLarsson-classicalSimulation}.

The aforementioned results show that studying general probabilistic theory is interesting even from practical viewpoint and that it could have potential applications in information processing.

\section{Introduction to general probabilistic theory} \label{sec:GPT}
General probabilistic theory is a unified framework to describe the kinematics of different systems in a mathematically unified fashion. The may idea of general probabilistic theory is an operational approach to setting the axioms and then carrying forward using convex analysis. Useful bookns on convex analysis are \cite{Rockafellar-convex, BoydVandenberghe-convex}. The beautiful aspect of general probabilistic theory is that it is only little bit more general than dealing with the different systems on their own, but we do not have to basically rewrite the same calculations over and over again for different theories.

During our calculations we will use two recurring examples, one will be finite-dimensional classical theory and other will be finite-dimensional quantum theory. The finite-dimensional classical theory is closely tied to the known results about incompatibility, steering and Bell nonlocality of measurements and we will mainly use it to verify that the definitions we will propose are, in the special case of measurements, the same as the known definitions. The quantum theory is our main concern as this is the theory we are mostly interested in. Some results, that we will only prove for quantum theory, may be generalized for general probabilistic theory, but we will limit the generality of our calculations to make them more understandable to readers that are not so far familiar with general probabilistic theory.

Given that we will work with many different spaces, their duals, their tensor products and many isomorphic sets, all isomorphisms will be omitted unless explicitly stated otherwise.

\subsection{The state space and the effect algebra of general probabilistic theory}
There are two central notions in general probabilistic theory: the state space that describes all possible states of the system and the effect algebra that describes the measurements on the system. We will begin our construction from the state space and then define the effect algebra, but we will show how one can go the other way and start from an effect algebra and obtain state space afterwards. We will restrict ourselves to finite-dimensional spaces and always use the Euclidean topology.

Let $V$ denote a real, finite-dimensional vector space and let $X \subset V$, then by $\conv(X)$ we will denote the convex hull of $X$, by $\aff(X)$ we will denote the affine hull of $X$. We will proceed with the definition of relative interior of a set $X \subset V$.
\begin{defin}
Let $X \subset V$, then the relative interior of $X$, denoted $\ri(X)$ is the interior of $X$ when it is considered as a subset of $\aff(X)$.
\end{defin}
For a more throughout discussion of relative interior see \cite[p. 44]{Rockafellar-convex}.

Let $K$ be a compact convex subset of $V$, then $K$ is a state space. The points $x \in K$ represent the states of some system and their convex combination is interpreted operationally, that is for $x, y \in K$, $\lambda \in [0, 1] \subset \mathbb{R}$ the state $\lambda x + (1-\lambda) y$ corresponds to having prepared $x$ with probability $\lambda$ and $y$ with probability $1-\lambda$.

To define measurements we have to be able to assign probabilities to states, that is we have to have a map $f: K \to [0, 1]$ such that, to follow the operational interpretation of convex combination, we have assign the convex combination of probabilities to the convex combination of respective states. In other words for $x, y \in K$, $\lambda \in [0, 1]$ we have to have
\begin{equation*}
f(\lambda x + (1-\lambda) y) = \lambda f(x) + (1-\lambda) f(y),
\end{equation*}
which means that $f$ is an affine function. Such functions are called effects because they correspond to assigning probabilities of measurement outcomes to states. We will proceed with a more formal definition of effects and of effect algebra.

Let $A(K)$ denote the set of affine functions $K \to \mathbb{R}$. $A(K)$ is itself a real linear space, moreover it is ordered as follows: let $f, g \in A(K)$, then $f \geq g$ if $f(x) \geq g(x)$ for every $x \in K$. There are two special functions $0$ and $1$ in $A(K)$, such that $0(x) = 0$ and $1(x) = 1$ for all $x \in K$.

The set $A(K)^+ = \{ f \in A(K): f \geq 0 \}$ is the convex, closed cone of positive functions. The cone $A(K)^+$ is generating, that is for every $f \in A(K)$ we have $f_{+}, f_{-} \in A(K)^+$ such that $f = f_{+} - f_{-}$, and it is pointed, that is if $f \geq 0$ and $-f \geq 0$, then $f = 0$.

Although we will provide a proper definition of measurement in subsection \ref{subsec:GPT-channels}, we will now introduce the concept of yes/no measurement, or two-outcome measurement, that will motivate the definition of the effect algebra. Our notion of measurement might seem different to the standard understanding and one may argue that what we will refer to as measurements are should be called entanglement-breaking maps, but this way of defining measurement is standard in general probabilistic theory, hence we will use it. A measurement is a procedure that assigns probabilities to possible outcomes based on the state that is measured. If we have only two outcomes and we know that the probability of the first outcome is $p \in [0 ,1]$, then, by normalization, the probability of the second outcome must be $1-p$. This shows that a two-outcome measurement needs to assign only probability to one outcome and the other probability follows.

Since assigning probabilities to states is a function $f: K \to [0, 1]$ and due to our operational interpretation of convex combination we want such function to be affine. Traditionally the functions that assign probabilities to states are called effects and the set of all effect is called effect algebra.
\begin{defin}
The set $E(K) = \{ f \in A(K): 0 \leq f \leq 1 \}$ is called the effect algebra.
\end{defin}
In general, one may define effect algebra in more general fashion, using the partially defined operation of addition and a unary operation $\perp$, that would in our case correspond to $f^\perp = 1 - f$, see \cite{FoulisBennett-effAlg} for a more throughout treatment.

Let $f \in E(K)$ then the two outcome measurement $m_f$ corresponding to the effect $f$ is the procedure that for $x \in K$ assigns the probability $f(x)$ to the first outcome and the probability $1-f(x)$ to the second outcome. Note that we did not mention any labels of the outcomes. Usually the outcomes are labeled yes and no, or $0$ and $1$, or $-1$ and $1$, but from operational perspective this does not matter.

We provide two standard examples of special cases of our definitions.
\begin{exm}[Classical theory]
In classical theory, the state space $K$ is a simplex, that is the convex hull of a set of affinely independent points $x_1, \ldots, x_n$. The special property of the simplex is that every point $x \in K$ can be uniquely expressed as convex combination of the points $x_1, \ldots, x_n$, due to their affine independence.
\end{exm}

\begin{exm}[Quantum theory]
Let $\Ha$ denote a finite-dimensional complex Hilbert space, let $B_h(\Ha)$ denote the real linear space of self-adjoint operators on $\Ha$, for $A \in B_h(\Ha)$ let $\Tr(A)$ denote the trace of the operator $A$ and let $A \geq 0$ denote that $A$ is positive semi-definite. We say that $A \leq B$ if $0 \leq B - A$. Let $B_h(\Ha)^+ = \{A \in B_h(\Ha): A \geq 0 \}$ denote the cone of positive semi-definite operators.

In quantum theory the state space is given as
\begin{equation*}
\dens_\Ha = \{ \rho \in B_h(\Ha): \rho \geq 0, \Tr(\rho) = 1 \}
\end{equation*}
which is the set of density operators on $\Ha$. The effect algebra $E(\dens_\Ha)$ is given as
\begin{equation*}
E(\dens_\Ha) = \{ M \in B_h(\Ha): 0 \leq M \leq \I \}.
\end{equation*}
The value of the effect $M \in E(\dens_\Ha)$ on the state $\rho \in \dens_\Ha$ is given as
\begin{equation*}
M(\rho) = \Tr( \rho M ).
\end{equation*}
\end{exm}

\subsection{The structure of general probabilistic theory}
This subsection will be rather technical, but we will introduce several mathematical results that we will use later on.

Let $x \in K$ and consider the map $\overline{x}: A(K) \to \mathbb{R}$, that to $f \in A(K)$ assigns the value $f(x)$. This is clearly a linear functional on $A(K)$. Moreover for $x, y \in K$, $\lambda \in [0, 1]$ we have
\begin{equation*}
\overline{\lambda x + (1-\lambda) y} = \lambda \overline{x} + (1-\lambda) \overline{y}
\end{equation*}
as the functions in $A(K)$ are affine by definition.  We conclude that the state space $K$ must be affinely isomorphic to some subset of the dual of $A(K)$. Since the aforementioned isomorphism is going to be extremely useful in later calculations we will describe it in more detail. Let $A(K)^*$ denote the dual of $A(K)$, that is the space of all linear functionals on $A(K)$. For $\psi \in A(K)^*$ and $f \in A(K)$ we will denote the value the functional $\psi$ reaches on $f$ as $\< \psi, f \>$. The dual cone to $A(K)^+$ is the cone $A(K)^{*+} = \{ \psi \in A(K)^*: \< \psi, f \> \geq 0, \forall f \in A(K)^+ \}$ that gives rise to the ordering on $A(K)^*$ given as follows: let $\psi, \varphi \in A(K)^*$, then $\psi \geq \varphi$ if and only if $(\psi - \varphi) \in A(K)^{*+}$, i.e. if $\psi - \varphi \geq 0$.

It is straightforward that the state space $K$ is isomorphic to a subset of the cone $A(K)^{*+}$, moreover it is straightforward to realize that the functionals isomorphic to $K$ must map the function $1 \in A(K)$ to the value $1$.

\begin{defin}
Let $\states_K = \{ \psi \in A(K)^{*+}: \< \psi, 1 \> = 1 \}$. We call $\states_K$ the state space of the effect algebra $E(K)$.
\end{defin}

It might be confusing at this point why we call $\states_K$ a state space, but this will be cleared by the following.

\begin{prop}
$\states_K$ is affinely isomorphic to $K$.
\end{prop}
\begin{proof}
It is clear that the map $x \to \overline{x}$ maps $K$ to a convex subset of $\states_K$. It is easy to show the inclusion of $\states_K$ in the image of $K$ using Hahn-Banach separation theorem, see \cite[Chapter 1, Theorem 4.3]{AsimowEllis} for a proof.
\end{proof}

We will omit the isomorphism between $K$ and $\states_K$, so for any $x, y \in K$, $\alpha \in \mathbb{R}$ we will write $\alpha x + y$ instead of $\alpha \overline{x} + \overline{y} \in A(K)^*$. Still, one must be careful when omitting this isomorphism, beacuse if $0 \in V$ denotes the zero vector and $0 \in K$, then $\overline{0} \in A(K)^*$ is not the zero functional as by construction we have $\< \overline{0}, 1 \> = 1$. We will do our best to avoid such possible problems by choosing appropriate notation.

There are two more result we will heavily rely on:
\begin{prop}
$\states_K$ is a base of $A(K)^{*+}$, that is for every $\psi \in A(K)^{*+}$, $\psi \neq 0$ there is a unique $x \in K$ and $\lambda \in \mathbb{R}$, $\lambda \geq 0$ such that $\psi = \lambda x$.
\end{prop}
\begin{proof}
Let $\psi \in A(K)^{*+}$, $\psi \neq 0$, then $\< \psi, 1 \> \neq 0$ as if $\< \psi, 1 \> = 0$ and $\psi \geq 0$, then $\psi  = 0$, because $1 \in \ri (A(K)^+)$. Let $\psi' = \frac{1}{\< \psi, 1 \>} \psi$. It is straightforward that $\psi' \in \states_K$.
\end{proof}

\begin{prop}
$A(K)^{*+}$ is a generating cone in $A(K)^*$, that is for every $\psi \in A(K)$ there are $\varphi_+, \varphi_{-} \in A(K)^{*+}$ such that $\psi = \varphi_+ - \varphi_-$.
\end{prop}
\begin{proof}
The result follows from the fact that $A(K)^+$ is a pointed cone, see \cite[Section 2.6.1]{BoydVandenberghe-convex}.
\end{proof}

\subsection{Tensor products of state spaces and effect algebras} \label{subsec:GPT-tensor}
Tensor products are a way to describe joint systems of several other systems. There are several approaches to introducing a tensor product in general probabilistic theory. There is a category theory based approach \cite{KissingerUiljen-category} that is a viable way to introduce the tensor products, but we will use a simpler, operational approach. Note that the state space of the joint system will be a compact convex subset of a real, finite-dimensional vector space as it itself must be a state space of some general probabilistic theory. Also keep in mind that describing a tensor product of state spaces $K_A$, $K_B$ is equivalent to describing the tensor product of the cones $A(K_A)^{*+}$, $A(K_B)^{*+}$. This is going to be useful as some things are easier to express in terms of the positive cones.

Let $V$, $W$ be real finite-dimensional vector spaces and let $v \in V$, $w \in W$. $v \otimes w$ will refer to the element of the algebraic tensor product $V \otimes W$, see e.g. \cite{Ryan-tensProd}. We will first describe the minimal and maximal tensor products of state spaces that set bounds on the real state space of the joint system. Note that when describing the joint state space of two state spaces or states of two systems, we will refer to them as bipartite state space or bipartite states.

Let $K_A$, $K_B$ denote two state spaces of Alice and Bob respectively. For every $x_A \in K_A$, $x_B \in K_B$ there must be a state of the joint system describing the situation that Alice's system is in the state $x_A$ and Bob's system is in the state $x_B$. We will denote such state $x_A \otimes x_B$ and we will call it a product state. Since the state space must be convex, the state space of the joint system must contain at least the convex hull of the product states. This leads to the definition of minimal tensor product.
\begin{defin}
The minimal tensor product of state spaces $K_A$ and $K_B$, denoted $K_A \tmin K_B$ is the compact convex set
\begin{equation*}
K_A \tmin K_B = \conv( \{ x_A \otimes x_B: x_A \in K_A, x_B \in K_B \} ).
\end{equation*}
\end{defin}
The bipartite states $y \in K_A \tmin K_B$ are also called separable states. For the positive cones we get
\begin{align*}
A(K_A \tmin K_B)^{*+} = \conv( \{ \psi_A \otimes \psi_B: &\psi_A \in A(K_A)^{*+}, \\
&\psi_B \in A(K_B)^{*+} \} ).
\end{align*}

\begin{exm}
In quantum theory, the minimal tensor product $\dens_\Ha \tmin \dens_\Ha$ is the set of all separable states, that is of all states of the form $\sum_{i=1}^n \lambda_i \rho_i \otimes \sigma_i$ for $n \in \mathbb{N}$ and $\rho_i \in \dens_\Ha$, $\sigma_i \in \dens_\Ha$, $0 \leq \lambda_i$ for $i \in \{1, \ldots, n\}$, $\sum_{i=1}^n \lambda_i = 1$.
\end{exm}

In a similar fashion, let $f_A \in E(K_A)$, $f_B \in E(K_B)$, then we can define a function $f_A \otimes f_B$ as the unique affine function such that for $x_A \in K_A$, $x_B \in K_B$ we have
\begin{equation*}
(f_A \otimes f_B) ( x_A \otimes x_B) = f_A(x_A) f_B(x_B).
\end{equation*}
This function is used in the most simple measurement on the joint system, such that Alice applies the two-outcome measurement $m_{f_A}$ and Bob applies the two outcome measurement $m_{f_B}$, so $f_A \otimes f_B$ must be an effect on the joint state space. This leads to the definition of the maximal tensor product.
\begin{defin}
The maximal tensor product of the state spaces $K_A$ and $K_B$, denoted $K_A \tmax K_B$, is defined as
\begin{align*}
K_A \tmax K_B = \{ &\psi \in A(K_A)^* \otimes A(K_B)^*: \forall f_A \in A(K_A)^+, \\
&\forall f_B \in A(K_B)^+, \< \psi, f_A \otimes f_B \> \geq 0 \} 
\end{align*}
\end{defin}
States in $K_A \tmax K_B \setminus K_A \tmin K_B$ are called entangled states. Equivalent definition, in terms of the positive cones would be
\begin{equation*}
A(K_A \tmax K_B)^{*+} = ( A(K_A)^{+} \tmin A(K_B)^{+} )^{*+}
\end{equation*}
where
\begin{align*}
A(K_A)^{+} \tmin A(K_B)^{+} = \conv( \{ f_A \otimes f_B: &f_A \in A(K_A)^+, \\
&f_B \in A(K_B)^+ \} ).
\end{align*}
As we see, the definition of tensor product of cones of positive functionals goes hand in hand with the definition of tensor product of cones of positive functions.

\begin{exm}
In quantum theory, the maximal tensor product of the cones $B_h(\Ha)^+ \tmax B_h(\Ha)^+$ is the cone of entanglement witnesses \cite[Section 6.3.1]{HeinosaariZiman-MLQT}, i.e. $W \in B_h(\Ha)^+ \tmax B_h(\Ha)^+$ if for every $\rho \in \dens_\Ha$, $\sigma \in \dens_\Ha$ we have $\Tr(W \rho \otimes \sigma) \geq 0$. Note that this does not imply the positivity of $W$.
\end{exm}

From the constructions it is clear that the state space of the joint system has to be a subset of the maximal tensor product and it has to contain the minimal tensor product. But there is no other specification of the state space of the joint system in general, it has to be provided by the theory we are working with.
\begin{defin}
We will call the joint state space of the systems described by the state spaces $K_A$ and $K_B$ the real tensor product of $K_A$ and $K_B$ and we will denote it $K_A \treal K_B$. We always have
\begin{equation*}
K_A \tmin K_B \subseteq K_A \treal K_B \subseteq K_A \tmax K_B.
\end{equation*}
\end{defin}

\begin{exm}
In quantum theory, the real tensor product of the state spaces is defined as the set of density matrices on the tensor product of the Hilbert spaces, that is
\begin{equation*}
\dens_\Ha \treal \dens_\Ha = \dens_{\Ha \otimes \Ha}.
\end{equation*}
\end{exm}

It is tricky to work with the tensor products in general probabilistic theory as the real tensor product is not always specified, or it may not be clear what it should be. We will always assume that every tensor product we need to be defined is defined. Moreover when working with a tensor product of more than two state spaces, say $K_A$, $K_B$, $K_C$ we will always assume that
\begin{equation*}
(K_A \treal K_B) \treal K_C = K_A \treal (K_B \treal K_C)
\end{equation*}
and we will simply write $K_A \treal K_B \treal K_C$. In the applications of general probabilistic theory to quantum and classical theory it will always be clear how to construct the needed tensor products and we consider this sufficient for us since we are mainly interested in the applications of our results.

We will state and prove a result about classical state spaces that we will use several times later on.
\begin{prop}
Let $S$ be a simplex with the extremal points $x_1, \ldots, x_n$, i.e. $S = \conv( \{ x_1, \ldots, x_n \} )$ and let $K$ be any state space, then we have
\begin{equation*}
S \tmin K = S \tmax K.
\end{equation*}
\end{prop}
\begin{proof}
Let $S$ be a simplex and let $x_i \in A(S)^{*+}$, $i \in \{ 1, \ldots, n \}$, be the extreme points of $S$. The points $x_1, \ldots, x_n$ form a basis of $A(S)^*$. Let $\psi \in S \tmax K$ then we have
\begin{equation*}
\psi = \sum_{i=1}^n x_i \otimes \varphi_i,
\end{equation*}
for some $\varphi_i \in A(K)^*$. Our aim is to prove that $\varphi_i \in A(K)^{*+}$ then $\psi \in S \tmin K$ follows by definition.

Let $b_1, \ldots, b_n$ denote the basis of $A(S)$ dual to the basis $x_1, \ldots, x_n$ of $A(S)^*$, i.e. we have $b_i(x_j) = \delta_{ij}$, where $i, j \in \{ 1, \ldots, n \}$ and $\delta_{ij}$ is the Kronecker delta. We have $b_i \in E(S)$ because $S$ is a simplex. For any $f \in E(K)$ we have
\begin{equation*}
0 \leq (\psi, b_i \otimes f) = (\varphi_i, f)
\end{equation*}
for all $i \in \{1, \ldots, n\}$, which implies $\varphi_i \in A(K)^{*+}$.
\end{proof}
Note that tensor product of the simplexes $S_1$, $S_2$ is also a simplex, so we have
\begin{equation*}
K \tmax S_1 \tmax S_2 = K \tmin S_1 \tmin S_2.
\end{equation*}

\subsection{Direct product of state spaces and effect algebras} \label{subsec:GPT-direct}
For certain reasons we will need to use direct products together with tensor product. The idea of why they will be used is going to be clear in the end, but now we will present several of their properties that will be required later. As in the Subsec. \ref{subsec:GPT-tensor} we will work mostly with the cones of the positive functionals.

Let $K_A, K_B$ be two state spaces. Given $A(K_A)^{*+}$ and $A(K_B)^{*+}$ there are two ways to define the direct product of these cones. The first is to use the cone $A(K_A)^{*+} \times A(K_B)^{*+}$. The second is to realize that we can construct $K_A \times K_B$ that will be a compact and convex set, i.e. a state space that gives rise to the cone $A(K_{B_1} \times K_{B_2})^{*+}$.

It may seem that these cones are fairly similar, but they are not and they have different physical interpretations. Let $\psi \in A(K_A \times K_B)^{*+}$, then there are unique $\lambda \in \mathbb{R}$, $x_A \in K_A$, $x_B \in K_B$ such that $\psi = \lambda (x_A, x_B)$. Now let $\varphi \in A(K_A)^{*+} \times A(K_B)^{*+}$, then there are $y_A \in K_A$, $y_B \in K_B$, $\alpha_A, \alpha_B \in \mathbb{R}$, $\alpha_A, \alpha_B \geq 0$ such that $\varphi = (\alpha_A y_A, \alpha_B y_B)$. In other words the normalization may be different in every component of the product. This can be rewritten as
\begin{align*}
\varphi &= (\alpha_A y_A, \alpha_B y_B) \\
&= (\alpha_A + \alpha_B) \left(\dfrac{\alpha_A}{\alpha_A + \alpha_B} y_A, \dfrac{\alpha_B}{\alpha_A + \alpha_B} y_B \right) \\
&= (\alpha_A + \alpha_B) \left( \dfrac{\alpha_A}{\alpha_A + \alpha_B} (y_A, 0) + \dfrac{\alpha_B}{\alpha_A + \alpha_B} ( 0, y_B ) \right)
\end{align*}
that shows that every element of $A(K_A)^{*+} \times A(K_B)^{*+}$ can be uniquely expressed as a multiple of a convex combination of elements of the form $(y_A, 0)$ and $(0, y_B)$. The operational interpretation of such states is that we do not even know which system we are working with, but we know that with some probability $p$ we have the first system and with probability $1-p$ we have the second system.

The operational interpretation of $A(K_A \times K_B)^{*+}$ is a bit harder to grasp. We may understand $\psi \in A(K_A \times K_B)^{*+}$ as a (multiple of) conditional state. That is, we will interpret the object $(x_A, x_B)$ as a state that corresponds to making a choice in the past between the systems $K_A$ and $K_B$ and keeping track of both of the outcomes at once. The cone $A(K_A \times K_B)^{*+}$ will play a central role in our results on incompatibility, steering and Bell nonlocality, because in the problem of incompatibility we wish to implement two channels at the same time and in steering and Bell nonlocality we are choosing between two incompatible channels.

At last we will need to describe the set $A(K_A \times K_B)$ and its structure with respect to the sets $A(K_A)$ and $A(K_B)$. We will show that $A(K_A \times K_B)$ corresponds to a certain subset of $A(K_A) \times A(K_B)$ by using the following two ideas: since all of the vector spaces are finite dimensional we have that $A(K_A) \times A(K_B)$ is the dual to $A(K_A)^* \times A(K_B)^*$ and $A(K_A \times K_B)^*$ can be identified with a subset of $A(K_A)^* \times A(K_B)^*$. Note that this identification holds only between the vector spaces and not between the corresponding state spaces.

\begin{prop}
We have
\begin{equation*}
A(K_{B_1} \times K_{B_2})^{*+} \subset A(K_{B_1})^{*+} \times A(K_{B_2})^{*+}.
\end{equation*}
\end{prop}
\begin{proof}
The idea of the proof is that if we have $\varphi \in A(K_A)^{*+} \times A(K_B)^{*+}$ such that $\varphi = (\alpha_A y_A, \alpha_B y_B)$ then $\varphi \in A(K_A \times K_B)^{*+}$ if and only if $\alpha_A = \alpha_B$. Therefore we can identify $A(K_A \times K_B)^{*+}$ with the set $\{ \psi \in A(K_A)^{*+} \times A(K_B)^{*+}: \<\psi, (1, -1) \> = 0 \}$. It is easy to verify this constraint on the positive cones and since it is linear it must hold everywhere else.
\end{proof}

The above proof shows that the function $(1, -1) \in A(K_A) \times A(K_B)$ is equal to zero when restricted to $A(K_A \times K_B)^*$, or in other words $(1, 0) = (0, 1)$ when restricted to $A(K_A \times K_B)^*$. We introduce a relation of equivalence on $A(K_A) \times A(K_B)$ as follows: for $f, g \in A(K_A) \times A(K_B)$ we say that $f$ and $g$ are equivalent and we write $f \sim g$ if $f - g = k (1, -1)$ for some $k \in \mathbb{R}$. Equivalently, $f \sim g$ if for every $\psi \in A(K_A \times K_B)^*$ we have $\< \psi, f \> = \<\psi, g \>$. $A(K_A \times K_B)$ corresponds to the set of equivalence classes of $A(K_A) \times A(K_B)$ with respect to the relation of equivalence $\sim$.

To demonstrate this, consider the constant function $1 \in E(K_A \times K_B)$ and let $x \in K_A$, $y \in K_B$, then we have
\begin{align*}
\< (x,y), (1,0) \> &= \< x, 1\> = 1 = \< (x,y), 1 \>, \\
\< (x,y), (0,1) \> &= \< y, 1\> = 1 = \< (x,y), 1 \>. \\
\end{align*}
This is not a coincidence, because $(1, 0) - (0, 1) = (1, -1)$ so we have $(1, 0) \sim (0, 1)$.

\subsection{Channels and measurements in general probabilistic theory} \label{subsec:GPT-channels}
It is not easy to define channels in general probabilistic theory as we would like all of the channels to be completely positive. We will use the following definition:
\begin{defin}
Let $K_A$, $K_B$ be state spaces, then channel $\Phi$ is a linear map
\begin{equation*}
\Phi: A(K_A)^* \to A(K_{B})^*
\end{equation*}
that is positive, i.e. for every $\psi \in A(K_A)^{*+}$ we have $\Phi(\psi) \in A(K_B)^{*+}$ and that for $\psi \in K_A$ we have $\Phi(\psi) \in K_B$.
\end{defin}

One may also require a channel to be completely positive, that is if $K_C$ is some state space such that we can define $K_C \treal K_A$, then we can consider the map $id \otimes \Phi : K_C \treal K_A \to K_C \tmax K_B$ and require it to be positive. In the applications of general probabilistic theory to classical and quantum theories, we always know how to create joint systems of given two systems so in the examples we will always require complete positivity of channels, but one still has to bear in mind that in the general case, complete positivity is not a well-defined concept.

One can identify the channel $\Phi: A(K_A)^* \to A(K_{B})^*$ with an element of $A(K_A) \otimes A(K_B)^*$ as follows: let $x \in K_A$ and $f \in A(K_B)$, then the expression $\< \Phi(x), f \>$ gives rise to a linear functional on $A(K_A)^* \otimes A(K_B)$. This means that we have $\Phi \in A(K_A) \otimes A(K_B)^*$, where we omit the isomorphism between the channel and the functional. If we also consider the positivity of the channel on the elements of the form $x \otimes f \in K_A \tmin E(K_B)$ we get
\begin{align*}
\Phi \in A(K_A)^+ \tmax A(K_B)^{*+}.
\end{align*}
This is a well known construction that may be also used to define the tensor product of linear spaces \cite[Chapter 1.3]{Ryan-tensProd}.

There is one more construction with channels that will be important in our formulation of compatibility of channels: compositions with effect. Let $\Phi: K_A \to K_B$ be a channel and let $f \in E(K_B)$, then they give rise to an effect $(f \circ \Phi) \in E(K_A)$ defined for $x_A \in K_A$ as
\begin{equation*}
\< x_A, (f \circ \Phi) \> = \< \Phi(x_A), f \>.
\end{equation*}
By the same idea we can define a map $f \otimes id: A(K_B)^* \otimes A(K_C)^* \to A(K_C)^*$ such that for $x_B \in K_B$ and $x_C \in K_C$ we have $(f \otimes id)(x_B \otimes x_C) = f(x_B)x_C$ and we extend the map by linearity. Also given a channel $\Phi: K_A \to K_B \treal K_C$ we can compose the map $f \otimes id$ with the channel $\Phi$ to obtain $(f \otimes id) \circ \Phi': A(K_A)^* \to A(K_C)^*$ such that the corresponding functional on $A(K_A) \otimes A(K_C)^*$ is for $x_A \in K_A$ and $g \in A(K_C)$ given as
\begin{equation*}
\< (f \otimes id) \circ \Phi, x_A \otimes g \> = \< \Phi(x_A), f \otimes g \>.
\end{equation*}
Specifically we will be interested in the expressions $(1 \otimes id) \circ \Phi$ and $(id \otimes 1) \circ \Phi$. If $\Phi$ is a channel then $(1 \otimes id) \circ \Phi$ and $(id \otimes 1) \circ \Phi$ are channels as well and they are called marginal channels of $\Phi$.

A special type of channel is a measurement.
\begin{defin}
A channel $m: K_A \to K_B$ is called a measurement if $K_B$ is a simplex.
\end{defin}
The interpretation is simple: the vertices of the simplex correspond to the possible measurement outcomes and the resulting state is a probability distribution over the measurement outcomes, i.e. an assignment of probabilities to the possible outcomes. Since we require all state spaces to be finite-dimensional this implies that we consider only finite-outcome measurements. Let $K_B$ be a simplex with vertices $\omega_1, \ldots, \omega_n$, then we can identify a measurement $m$ with an element of $A(K_A)^+ \tmin A(K_B)^{*+}$ of the form
\begin{equation*}
m = \sum_{i=1}^n f_i \otimes \delta_{\omega_i}
\end{equation*}
where for $i \in \{1, \ldots, n\}$ we have $f_i \in E(K_A)$, $\sum_{i=1}^n f_i = 1$ and $\delta_{\omega_i} \in \states (K_B)$ are the functionals corresponding to the extreme points of $K_B$ (where we have not omitted the isomorphism this time). This expression has an operational interpretation that for $x \in K_A$ the measurement $m$ assigns the probability $f_i(x)$ to the outcome $\omega_i$.

\begin{exm}
Quantum channels are completely positive, trace preserving maps $\Phi: \dens_\Ha \to \dens_\Ha$. The complete positivity means that for any $\rho \geq 0$ we have $(id \otimes \Phi)(\rho) \geq 0$. We denote the set of channels $\Phi: \dens_\Ha \to \dens_\Ha$ as $\chan_{\Ha \to \Ha}$.

Let $|1\>, \ldots, |n\>$ be an orthonormal base of $\Ha$. To every quantum channel we may assign its unique Choi matrix $C(\Phi)$ defined as
\begin{equation*}
C(\Phi) = (\Phi \otimes id) \left( \sum_{i, j=1}^n |ii\>\<jj| \right) ,
\end{equation*}
where we use the shorthand $|i i \> = |i\> \otimes |i\>$. Note that $C(\Phi) \geq 0$ and $\Tr_1 (C(\Phi)) = \I$, where $\Tr_1$ denotes the partial trace. Also every matric $C \in B_h(\Ha \otimes \Ha)$ such that $C \geq 0$ and $\Tr_1(C) = \I$ is a Choi matrix of some channel, see \cite[Section 4.4.3]{HeinosaariZiman-MLQT}.

The Choi matrix $C(\Phi)$ is isomorphic to a state $\frac{1}{\dH} C(\Phi)$, which corresponds to the channel $\Phi \otimes id$ acting on the maximally entangled state $|\psi^+\>\<\psi^+|$, where
\begin{equation*}
|\psi^+\> = \dfrac{1}{\sqrt{\dH}} \sum_{i=1}^n |i i \>.
\end{equation*}
\end{exm}

\section{Compatibility of channels} \label{sec:compat}
\begin{defin}
Let $K_A$, $K_{B_1}$, $K_{B_2}$ be state spaces and let $\Phi_1$, $\Phi_2$ be channels
\begin{align*}
&\Phi_1: K_A \to K_{B_1}, \\
&\Phi_2: K_A \to K_{B_2}.
\end{align*}
We say that $\Phi_1$, $\Phi_2$ are compatible if and only if there exists a channel
\begin{equation*}
\Phi : K_A \to K_{B_1} \treal K_{B_2}
\end{equation*}
such that $\Phi_1$ and $\Phi_2$ are the marginal channels of $\Phi$, i.e. we have
\begin{align}
\Phi_1 = (id \otimes 1) \circ \Phi, \label{eq:compat-def-1} \\
\Phi_2 = (1 \otimes id) \circ \Phi. \label{eq:compat-def-2}
\end{align}
The channel $\Phi$ is also called the joint channel of the channels $\Phi_1$, $\Phi_2$.
\end{defin}

The operational meaning of compatibility of channels is that if the channels $\Phi_1$, $\Phi_2$ are compatible, then we can apply them both to the input state at once and selecting which one we actually want the output from later. If the channels are incompatible we have to choose from which one we want the output before applying anything. For a more in-depth explanation see \cite{HeinosaariMiyaderaZiman-compatibility}. The important thing is that there is a choice from which channel we want to get the output so we can expect to see $A(K_{B_1} \times K_{B_2})^{*+}$ come up in the calculations.

Consider the channel $\Phi : K_A \to K_{B_1} \treal K_{B_2}$. One can realize that the maps $(id \otimes 1): \Phi \mapsto (id \otimes 1) \circ \Phi$ and $(1 \otimes id): \Phi \mapsto (1 \otimes id) \circ \Phi$ are linear maps of channels. Moreover the Eq. \eqref{eq:compat-def-1}, \eqref{eq:compat-def-2} both have $\Phi$ on the right hand side in the same position. We are going to exploit this to obtain simpler condition for compatibility of the channels $\Phi_1$, $\Phi_2$. To do so we have to introduce a new map $J$.

Let us define a map $J: A(K_A) \otimes A(K_{B_1})^* \otimes A(K_{B_2})^* \to A(K_A) \otimes A(K_{B_1} \times K_{B_2})^*$ given for $\Xi \in A(K_A) \otimes A(K_{B_1})^* \otimes A(K_{B_2})^*$ as
\begin{equation*}
J(\Xi) = \left( (id \otimes 1) \circ \Xi, (1 \otimes id) \circ \Xi \right).
\end{equation*}
For $\Xi = f \otimes \psi \otimes \varphi$ we have
\begin{equation*}
J(\Xi) = f \otimes ( \< \varphi, 1 \> \psi, \< \psi, 1 \> \varphi ).
\end{equation*}

\begin{prop} \label{prop:compat-Jlinear}
$J$ is a linear mapping.
\end{prop}
\begin{proof}
Let $\Xi_1, \Xi_2 \in A(K_A) \otimes A(K_{B_1} \otimes K_{B_2})^*$ and $\lambda \in \mathbb{R}$, then we have
\begin{align*}
J(\lambda \Xi_1 + \Xi_2) &= \left( \lambda (id \otimes 1) \circ \Xi_1 + (id \otimes 1) \circ \Xi_2, 0 \right) \\
&+ \left( 0, \lambda (1 \otimes id) \circ \Xi_1 + (1 \otimes id) \circ \Xi_2 \right) \\
&= \lambda \left( (id \otimes 1) \circ \Xi_1, (1 \otimes id) \circ \Xi_1 \right) \\
&+ \left( (id \otimes 1) \circ \Xi_2, (1 \otimes id) \circ \Xi_2 \right) \\
&= \lambda J(\Xi_1) + J(\Xi_2).
\end{align*}
\end{proof}

Assume that the channels $\Phi_1$, $\Phi_2$ are compatible and that $\Phi$ is their joint channel then we must have
\begin{equation*}
J(\Phi) = (\Phi_1, \Phi_2)
\end{equation*}
which is just a more compact form of the Eq. \eqref{eq:compat-def-1}, \eqref{eq:compat-def-2}.
\begin{prop} \label{prop:compat-condition}
The channels $\Phi_1$, $\Phi_2$ are compatible if and only if there is $\Phi \in A(K_A)^+ \tmax A(K_{B_1} \treal K_{B_2})^{*+}$ such that
\begin{equation}
J({\Phi}) = (\Phi_1, \Phi_2). \label{eq:compat-Jcond}
\end{equation}
\end{prop}
\begin{proof}
If the channels $\Phi_1$, $\Phi_2$ are compatible then Eq. \eqref{eq:compat-Jcond} must hold for their joint channel $\Phi$. If Eq. \eqref{eq:compat-Jcond} holds for some $\Phi \in A(K_A)^+ \tmax A(K_{B_1} \treal K_{B_2})^{*+}$, then the channels $\Phi_1$, $\Phi_2$ are compatible and $\Phi$ is their joint channel.
\end{proof}

The operational interpretation is that $(\Phi_1, \Phi_2)$ represents a conditional channel in the same way as the states from $A(K_{B_1} \times K_{B_2})^{*+}$ represent conditional states that keep track of some choice made in the past. If the channels are compatible, then we actually do not have to make the choice of either using $\Phi_1$ or $\Phi_2$, but we can use their joint channel, that has the property that its marginals reproduce the outcomes of the two channels $\Phi_1$, $\Phi_2$. We will investigate several of the properties of the map $J$.
\begin{prop} \label{prop:compat-JpreimageExists}
For every $(\xi_1, \xi_2) \in A(K_A) \otimes A(K_{B_1} \times K_{B_2})^*$ there is a $\Xi \in A(K_A) \otimes A(K_{B_1})^* \otimes A(K_{B_2})^*$ such that
\begin{equation*}
J(\Xi) = (\xi_1, \xi_2).
\end{equation*}
Moreover if we have
\begin{equation*}
(1, 1) \circ (\xi_1, \xi_2) = 1
\end{equation*}
then
\begin{equation*}
(1 \otimes 1) \circ \Xi = 1.
\end{equation*}
\end{prop}
\begin{proof}
Let $f_1, \ldots, f_n$ be a basis of $A(K_A)$, then we have
\begin{align*}
\xi_1 &= \sum_{i=1}^n f_i \otimes \psi_i \\
\xi_2 &= \sum_{i=1}^n f_i \otimes \varphi_i
\end{align*}
for some $\psi_i \in A(K_{B_1})^*$ and $\varphi_i \in A(K_{B_2})^*$. Since we must have
\begin{equation*}
(1,0) \circ (\xi_1, \xi_2) = (0,1) \circ (\xi_1, \xi_2)
\end{equation*}
we obtain
\begin{align*}
\sum_{i=1}^n \<\psi_i, 1\> f_i = \sum_{i=1}^n \<\varphi_i, 1\> f_i
\end{align*}
which implies
\begin{equation*}
\<\psi_i, 1 \> = \<\varphi_i, 1 \> = k_i
\end{equation*}
for all $i \in \{1, \ldots, n \}$ as $f_1, \ldots, f_n$ is linearly independent. Let
\begin{equation*}
\Xi = \sum_{i=1}^n k_i^{-1} f_i \otimes \psi_i \otimes \varphi_i
\end{equation*}
then we have
\begin{align*}
J(\Xi) &= \sum_{i=1}^n k_i^{-1} f_i \otimes ( \< \varphi_i, 1_{B_2} \> \psi_i, \< \psi_i, 1_{B_1} \> \varphi_i ) \\
&= \sum_{i=1}^n f_i \otimes (\psi_i, \varphi_i ).
\end{align*}
If we have $1 \circ (\xi_1, \xi_2) = 1$ then
\begin{equation*}
\sum_{i=1}^n k_i f_i = 1
\end{equation*}
and we get
\begin{align*}
(1 \otimes 1) \circ \Xi &= (1 \otimes 1) \circ ( \sum_{i=1}^n k_i^{-1} f_i \otimes \psi_i \otimes \varphi_i ) \\
&= \sum_{i=1}^n k_i^{-1} \< \psi_i, 1 \> \< \varphi_i, 1 \> f_i = 1.
\end{align*}
\end{proof}

\begin{prop} \label{prop:compat-JminimalProd}
We have
\begin{align*}
J(A(K_A)^+ \tmin A(K_{B_1})^{*+} \tmin A(K_{B_2})^{*+}) = \\
= A(K_A)^+ \tmin A(K_1 \times K_{B_2})^{*+}.
\end{align*}
\end{prop}
\begin{proof}
Let $(\xi_1, \xi_2) \in A(K_A)^+ \tmin A(K_{B_1} \times K_{B_2})^{*+}$ then as in the proof of Prop. \ref{prop:compat-JpreimageExists} we have
\begin{align*}
\xi_1 &= \sum_{i=1}^n f_i \otimes \psi_i \\
\xi_2 &= \sum_{i=1}^n f_i \otimes \varphi_i
\end{align*}
but now we have $f_i \geq 0$, $\psi_i \geq 0$ and $\varphi_i \geq 0$ for $i \in \{ 1, \ldots, n \}$. It follows by the same construction as in the proof of Prop. \ref{prop:compat-JpreimageExists} that  we can construct $\Xi = \sum_{i=1}^n k_i^{-1} f_i \otimes \psi_i \otimes \varphi_i$ and we get $\Xi \in A(K_A)^+ \tmin A(K_{B_1})^{*+} \tmin A(K_{B_2})^{*+}$.

Let $\Xi \in A(K_A)^+ \tmin A(K_{B_1})^{*+} \tmin A(K_{B_2})^{*+}$, then we have $\Xi = \sum_{i=1}^n f_i \otimes \psi_i \otimes \varphi_i$ such that $f_i \geq 0$, $\psi_i \geq 0$, $\varphi_i \geq 0$ for all $i \in \{1, \ldots, n \}$, moreover without lack of generality we can assume $\< \psi_i, 1_{B_1} \> = \< \varphi_i, 1_{B_2} \> = 1$. We have
\begin{align*}
J(\Xi) = \sum_{i=1}^n f_i \otimes (\psi_i, \varphi_i) \in A(K_A)^+ \tmin A(K_{B_1} \times K_{B_2})^{*+}
\end{align*}
which concludes the proof.
\end{proof}

It would be very useful to know what is the image of the cone $A(K_A)^+ \tmax A(K_{B_1} \treal K_{B_2})^{*+}$ when mapped by $J$. We will denote the resulting cone $Q = J(A(K_A)^+ \tmax A(K_{B_1} \treal K_{B_2})^{*+})$. The cone is important due to the following:
\begin{coro} \label{coro:compat-iff}
The channels $\Phi_1$, $\Phi_2$ are compatible if and only if
\begin{equation*}
(\Phi_1, \Phi_2) \in Q = J(A(K_A)^+ \tmax A(K_{B_1} \treal K_{B_2})^{*+}).
\end{equation*}
\end{coro}
\begin{proof}
Follows from Prop. \ref{prop:compat-condition}.
\end{proof}

\begin{prop}
$A(K_A)^+ \tmin A(K_{B_1} \times K_{B_2})^{*+} \subset Q$.
\end{prop}
\begin{proof}
Since
\begin{align*}
A(K_A)^+ \tmin A(K_{B_1} \tmin K_{B_2})^{*+} \subset A(K_A)^+ \tmax A(K_{B_1} \treal K_{B_2})^{*+}
\end{align*}
we must have
\begin{equation*}
J(A(K_A)^+ \tmin A(K_{B_1} \tmin K_{B_2})^{*+}) \subset Q.
\end{equation*}
The result follows from Prop. \ref{prop:compat-JminimalProd}.
\end{proof}

\begin{prop}
$Q \subset A(K_A)^+ \tmax A(K_{B_1} \times K_{B_2})^{*+}$.
\end{prop}
\begin{proof}
Since we have
\begin{align*}
A(K_A)^+ \tmax A(K_{B_1} \treal K_{B_2})^{*+} \subset A(K_A)^+ \tmax A(K_{B_1} \tmax K_{B_2})^{*+}
\end{align*}
we must have
\begin{equation*}
Q \subset J(A(K_A)^+ \tmax A(K_{B_1} \tmax K_{B_2})^{*+}).
\end{equation*}
Let $\Xi \in A(K_A)^+ \tmax A(K_{B_1} \tmax K_{B_2})^{*+}$, then for $\psi \in A(K_A)^{*+}$ and $(f_1, f_2) \in A(K_{B_1} \times K_{B_2})^+$ we get
\begin{align*}
\< J(\Xi), x \otimes f \> &= \< \left( (id \otimes 1) \circ \Xi, (1 \otimes id) \circ \Xi \right), x \otimes f \> \\
&= \< \Xi(x), f_1 \otimes 1 \> + \< \Xi(x), 1 \otimes f_2 \> \geq 0,
\end{align*}
that shows we have $J(A(K_A)^+ \tmax A(K_{B_1} \tmax K_{B_2})^{*+}) \subset A(K_A)^+ \tmax A(K_{B_1} \times K_{B_2})^{*+}$ which concludes the proof.
\end{proof}

We can also construct $Q$ as the cone we get when we factorize the cone $A(K_A)^+ \tmax  A(K_{B_1} \treal K_{B_2})^{*+} $ with respect to the relation of equivalence given as follows: $\Xi_1 \approx \Xi_2$ if and only if $J(\Xi_1) = J(\Xi_2)$, or equivalently if and only if $\Xi_1 = \Xi_2 + \Xi$, such that $J(\Xi) = 0$.

Note that since $J$ is a linear map, as we showed in Prop. \ref{prop:compat-Jlinear}, it is clear that $Q$ is a convex cone. For two given channels $\Phi_1: K_A \to K_{B_1}$, $\Phi_2: K_A \to K_{B_1}$ one may write a primal linear program that would check the condition for compatibility given by Cor. \ref{coro:compat-iff}. We will write such linear program for quantum channels later.

\section{Compatibility of measurements} \label{sec:meascompat}
We will apply the results of Sec. \ref{sec:compat} to the problem of compatibility of measurements. We will obtain the same results that were recently presented in \cite{Jencova-compatibility}, that are generalization a of \cite{NamiokaPhelps-tensorProd}.

Let $K_A$ be a state space and let $S_1$, $S_2$ be simplexes and let $m_1: K_A \to S_1$, $m_2: K_A \to S_2$  be measurements. According to Prop. \ref{prop:compat-condition} the measurements $m_1$, $m_2$ are compatible if and only if
\begin{equation*}
(m_1, m_2) \in J(A(K_A)^+ \tmax A(S_1 \treal S_2)^{*+}).
\end{equation*}
Since both $S_1$ and $S_2$ are simplexes, then we have $S_1 \treal S_2 = S_1 \tmin S_2$ and the condition for compatibility reduces according to Prop. \ref{prop:compat-JminimalProd} to
\begin{equation*}
(m_1, m_2) \in A(K_A)^+ \tmin A(S_1 \times S_2)^{*+}.
\end{equation*}

Due to the simpler structure of simplexes one may get even more specific results about measurements, see \cite{Jencova-compatibility}.

For demonstration of the derived conditions we will reconstruct the result of \cite{NamiokaPhelps-tensorProd} about compatibility of two-outcome measurements. According to our definition, a measurement is two-outcome if the simplex it has as a target space has two vertexes, i.e. it is a line segment. Let $K$ be a state space, $f, g \in E(K)$ and $m_f: K \to S$, $m_g: K \to S$ be two-outcome measurements given as
\begin{align*}
m_f &= f \otimes \delta_{\omega_1} + (1-f) \otimes \delta_{\omega_2}, \\
m_g &= g \otimes \delta_{\omega_1} + (1-g) \otimes \delta_{\omega_2}.
\end{align*}
The state space given by $A(S \times S)^{*+}$ is a square given as $\conv( (\delta_{\omega_1}, \delta_{\omega_1}), (\delta_{\omega_1}, \delta_{\omega_2}), (\delta_{\omega_2}, \delta_{\omega_1}), (\delta_{\omega_2}, \delta_{\omega_2}) )$, that is just affinely isomorphic to $S \times S$. We have
\begin{align*}
(m_1, m_2) &= f \otimes (\delta_{\omega_1}, 0 ) + (1-f) \otimes (\delta_{\omega_2}, 0 ) \\
&+ g \otimes (0, \delta_{\omega_1}) + (1-g) \otimes (0, \delta_{\omega_2}) \\
&= f \otimes (\delta_{\omega_1}, \delta_{\omega_2}) + (1-f) \otimes (\delta_{\omega_2}, \delta_{\omega_2}) \\
&+ g \otimes (0, \delta_{\omega_1} - \delta_{\omega_2}),
\end{align*}
where in the second step we have used the basis $(\delta_{\omega_1}, \delta_{\omega_2})$, $(\delta_{\omega_2}, \delta_{\omega_2})$, $(0, \delta_{\omega_1} - \delta_{\omega_2})$ of $A(S\times S)^*$ to express $(m_1, m_2)$ in a more reasonable form. To have $(m_1, m_2) \in A(K)^+ \tmin A(S \times S)^{*+}$ we must have
\begin{align*}
(m_1, m_2) &= h_{11} \otimes (\delta_{\omega_1}, \delta_{\omega_1}) + h_{12} \otimes (\delta_{\omega_1}, \delta_{\omega_2}) \\
&+ h_{21} \otimes (\delta_{\omega_2}, \delta_{\omega_1}) + h_{22} \otimes (\delta_{\omega_2}, \delta_{\omega_2}) \\
&= ( h_{11} + h_{12} ) \otimes (\delta_{\omega_1}, \delta_{\omega_2}) \\
&+ ( h_{21} + h_{22} ) \otimes (\delta_{\omega_2}, \delta_{\omega_2}) \\
&+ ( h_{11} + h_{21} ) \otimes (0, \delta_{\omega_1} - \delta_{\omega_2}),
\end{align*}
for some $h_{11}, h_{12}, h_{21}, h_{22} \in E(K)$. This implies the standard conditions for the compatibility of two-outcome measurements $m_f$, $m_g$:
\begin{align*}
f &= h_{11} + h_{12}, \\
1-f &= h_{21} + h_{22}, \\
g &= h_{11} + h_{21},
\end{align*}
see e.g. \cite{Plavala-simplex}.

\section{Compatibility of quantum channels} \label{sec:quantcompat}
In this section we will derive results more specific to the compatibility of quantum channels. Let $\Phi_1: \dens_\Ha \to \dens_\Ha$, $\Phi_2: \dens_\Ha \to \dens_\Ha$ be quantum channels, then according to Prop. \ref{prop:compat-condition} they are compatible if and only if there is a channel $\Phi: \dens_\Ha \to \dens_{\Ha \otimes \Ha}$ such that for all $\rho \in \dens_\Ha$ we have
\begin{equation}
(\Phi_1 (\rho), \Phi_2(\rho) ) = ( \Tr_2( \Phi (\rho) ), \Tr_1 ( \Phi ( \rho ) ) ). \label{eq:quantcompat-baseCond}
\end{equation}
This is equivalent to the definition of compatibility of quantum channels already stated in \cite{HeinosaariMiyadera-compOfChan}. It is straightforward that Eq. \eqref{eq:quantcompat-baseCond} implies that
\begin{equation*}
( C(\Phi_1), C(\Phi_2) ) = ( \Tr_2( C(\Phi) ), \Tr_1 ( C(\Phi) ) ),
\end{equation*}
we will show that they are equivalent. This will help us to get rid of the state $\rho$ in Eq. \eqref{eq:quantcompat-baseCond}.

\begin{prop} \label{prop:quantcompat-Choi}
The channels $\Phi_1: \dens_\Ha \to \dens_\Ha$, $\Phi_2: \dens_\Ha \to \dens_\Ha$ are compatible if and only if there exists a channel $\Phi: \dens_\Ha \to \dens_{\Ha \otimes \Ha}$ such that
\begin{equation*}
( C(\Phi_1), C(\Phi_2) ) = ( \Tr_2( C(\Phi) ), \Tr_1 ( C(\Phi) ) ).
\end{equation*}
\end{prop}
\begin{proof}
Let $\rho \in \dens_\Ha$, then we have
\begin{align*}
\Tr_2( \Phi (\rho) ) &= \Tr_{2, E} ( C(\Phi) \I \otimes \I \otimes \rho^T ) \\
&= \Tr_{E} ( \Tr_2 ( C(\Phi) ) \I \otimes \rho^T ) \\
&= \Tr_E ( C(\Phi_1) \I \otimes \rho^T ) = \Phi_1 (\rho).
\end{align*}
The same follows for $\Phi_2$.
\end{proof}

As we already showed in Sec. \ref{sec:compat}, the cone $Q = J( A(\dens_\Ha)^+ \tmax A( \dens_{\Ha \otimes \Ha})^{*+})$ is of interest for the compatibility of channels. In the case of quantum channels we will use Prop. \ref{prop:quantcompat-Choi} to formulate similar cone in terms of Choi matrices of the channels and we will write a semi-definite program for the compatibility of quantum channels based on this approach.

Denote $P = \{ (\Tr_2 (C), \Tr_1 (C)) : C \in \chan_{\Ha \to \Ha \otimes \Ha} \}$, then according to Prop. \ref{prop:quantcompat-Choi} the channels $\Phi_1: \dens_\Ha \to \dens_\Ha$, $\Phi_2: \dens_\Ha \to \dens_\Ha$ are compatible if and only if
\begin{equation*}
( C(\Phi_1), C(\Phi_2) ) \in P.
\end{equation*}
Note that, by our definition, $P$ is not a cone, but it generates some cone just by adding all of the operators of the form $\lambda C$, where $C \in P$ and $\lambda \in \mathbb{R}$, $\lambda \geq 0$.

It would be very interesting to obtain more specific results on the structure of $P$, but the task is not trivial. To make it simpler we will investigate the structure of the dual cone $P^*$ given as
\begin{align*}
P^* = \{& (A, B) \in B_h(\Ha) \times B_h(\Ha): \\
&\< C, (A, B) \> \geq 0, \; \forall C \in P \}.
\end{align*}
Notice that $(A, B) \in B_h(\Ha) \times B_h(\Ha)$ is simply a block-diagonal matrix having blocks $A$ and $B$. Also every $C \in P$ is a block diagonal matrix, let $C = (C_1, C_2)$, then
\begin{equation*}
\< (C_1, C_2), (A, B) \> = \Tr(C_1 A) + \Tr(C_2 B).
\end{equation*}
Let $C \in P$, then by definition there exist a channel $\Phi: \dens_\Ha \to \dens_{\Ha \otimes \Ha}$ such that
\begin{equation*}
C = ( \Tr_2 (C(\Phi)), \Tr_1 ( C(\Phi)) ).
\end{equation*}
Let $(A, B) \in P^*$, then we have
\begin{align*}
\< C, (A, B) \> &= \Tr \big( \Tr_2 (C(\Phi)) A + \Tr_1 ( C(\Phi)) B \big) \\
&= \Tr ( C(\Phi) (\tilde{A} + \I \otimes B ) ) \geq 0,
\end{align*}
where $\tilde{A}$ is the operator such that $\Tr ( \Tr_2 (C(\Phi)) A ) = \Tr ( C(\Phi) \tilde{A} )$. If $A = A_1 \otimes A_2$, then $\tilde{A} = A_1 \otimes \I \otimes A_2$. In general one can write $A$ as a sum of factorized operators and express $\tilde{A}$ in such way, because the map $A \mapsto \tilde{A}$ is linear.

The result is that $\tilde{A} + \I \otimes B$ must correspond to a positive function on quantum channels, hence we must have $\tilde{A} + \I \otimes B \geq 0$, see \cite{Jencova-extremalGenMeas, Ziman-ppovm}. We have proved the following:
\begin{prop} \label{prop:quantcompat-dualCone}
The channels $\Phi_1: \dens_\Ha \to \dens_\Ha$, $\Phi_2: \dens_\Ha \to \dens_\Ha$ are compatible if and only if
\begin{equation*}
\Tr( C(\Phi_1) A ) + \Tr( C(\Phi_2) B ) \geq 0
\end{equation*}
for all $A, B \in B_h(\Ha \otimes \Ha)$ such that
\begin{equation*}
\tilde{A} + \I \otimes B \geq 0.
\end{equation*}
\end{prop}

This allows us to formulate the semi-definite program \cite{BoydVandenberghe-convex} for the compatibility of quantum channels as follows:
\begin{prop}
Given channels $\Phi_1: \dens_\Ha \to \dens_\Ha$, $\Phi_2: \dens_\Ha \to \dens_\Ha$, the semi-definite program for the compatibility of quantum channels is
\begin{align*}
\inf_{A, B} \Tr( C(\Phi_1) A ) + \Tr( C(\Phi_2) B ) \\
\tilde{A} + \I \otimes B \geq 0,
\end{align*}
where $\tilde{A}$ is given as above.

If the reached infimum is negative, then the channels are incompatible, if the reached infimum is $0$ then the channels are compatible.
\end{prop}
\begin{proof}
The result follows from Prop. \ref{prop:quantcompat-dualCone}. One may see that the infimum is at most $0$ because one may always chose $A = B = 0$.
\end{proof}

\section{Prelude to steering and Bell nonlocality} \label{sec:prelude}
We will propose a possible test for the compatibility of channels that will not work, but it will motivate our definitions of steering and Bell nonlocality.

Let $K_A, K_{B_1}, K_{B_2}$ be state spaces and let $\Phi_1: K_A \to K_{B_1}$, $\Phi_2: K_A \to K_{B_2}$ be channels. The channels $\Phi_1$, $\Phi_2$ are compatible if Eq. \eqref{eq:compat-Jcond} is satisfied for some channel $\Phi: K_A \to K_{B_1} \treal K_{B_2}$. This is the same as saying the channels $\Phi_1$, $\Phi_2$ are compatible if for all $x \in K_A$ we have
\begin{equation}
(\Phi_1(x), \Phi_2(x)) = ( ((id \otimes 1) \circ \Phi)(x), ((1 \otimes id) \circ \Phi)(x) ). \label{eq:prelude-compat}
\end{equation}
If the channels $\Phi_1$ and $\Phi_2$ are compatible, then for every $x \in K_A$ there must exist a state $y \in K_{B_1} \treal K_{B_2}$ such that
\begin{align}
\Phi_1(x) &= (id \otimes 1)(y), \label{eq:prelude-x1} \\
\Phi_2(x) &= (1 \otimes id)(y). \label{eq:prelude-x2}
\end{align}
Would it be a reasonable test for the compatibility of the channels $\Phi_1$ and $\Phi_2$ if we considered the state $x \in K_A$ fixed and we would test whether, for the fixed state $x$, there exists $y \in K_{B_1} \treal K_{B_2}$ such that Eg. \eqref{eq:prelude-x1}, \eqref{eq:prelude-x2} are satisfied? It would not, because for a fixed $x \in K_A$ one always has $\Phi_1 (x) \otimes \Phi_2 (x) \in K_{B_1} \treal K_{B_2}$ that satisfies Eg. \eqref{eq:prelude-x1}, \eqref{eq:prelude-x2}.

Still, throwing away this line of thinking would not be a good choice, because going further, on may ask: if there would be another system $K_C$, such that $K_C \treal K_A$ is defined, then what if we would use the entanglement between the systems $K_A$ and $K_C$ to obtain a better condition for the compatibility of the channels $\Phi_1$, $\Phi_2$ using the very same line of thinking? As we will see, this approach leads to the notions of steering and Bell nonlocality.

\section{Steering} \label{sec:steering}
Steering is one of the puzzling phenomena we find in quantum theory but not in classical theory. It is usually described as a two party protocol, that allows one side to alter the state of the other in a way that would not be possible in classical theory by performing a measurement and announcing the outcome. Although originally discovered by Schr\"{o}dinger \cite{Schrodinger-steering}, steering was formalised in \cite{WisemanJonesDoherty-nonlocal}. Recently there was introduced a new formalism for steering in \cite{Jencova-compatibility}.

So far it was always only considered that during steering one party performs a measurement. Since a measurement is a special case of a channel, one may ask whether it is possible to define steering by channels. We will use our formalism for compatibility of channels to introduce steering by channels by continuing the line of thoughts presented in Sec. \ref{sec:prelude}. We will have to formulate steering in a little different way than it usually is formulated for measurements, but we will show that for measurements we will obtain the known results.

Let $K_A, K_{B_1}, K_{B_2}, K_C$ be finite-dimensional state spaces, such that $K_C \treal K_A$ is defined and let
\begin{align*}
\Phi_1&: K_A \to K_{B_1}, \\
\Phi_2&: K_A \to K_{B_2},
\end{align*}
be channels. We can construct channels
\begin{align*}
id \otimes \Phi_1 &: A(K_C)^{*+} \treal A(K_A)^{*+} \to A(K_C)^{*+} \tmax A(K_{B_1})^{*+}, \\
id \otimes \Phi_2 &: A(K_C)^{*+} \treal A(K_A)^{*+} \to A(K_C)^{*+} \tmax A(K_{B_2})^{*+}.
\end{align*}
Moreover we can construct the conditional channel
\begin{align*}
id \otimes (\Phi_1, \Phi_2)&: A(K_C)^{*+} \treal A(K_A)^{*+} \to \\
&\to A(K_C)^{*+} \tmax A(K_{B_1} \times K_{B_2})^{*+}.
\end{align*}
These channels play a central role in steering and we will keep this notation throughout this section. First, we will introduce a handy name for the output state of $id \otimes (\Phi_1, \Phi_2)$.

\begin{defin}
Let $\psi \in K_C \treal K_A$ be a bipartite state, then we call $(id \otimes (\Phi_1, \Phi_2))(\psi)$ a bipartite conditional state.
\end{defin}

Steering may be seen as a three party protocol that tests the compatibility of channels. The parties in question will be named Alice, Bob and Charlie. Alice and Charlie share a bipartite state $\psi \in K_C \treal K_A$ and Alice has the channels $\Phi_1$ and $\Phi_2$ at her disposal, that would send her part of the state $\psi$ to Bob. Since Alice can choose between the channels $\Phi_1$ and $\Phi_2$, she will be, in our formalism, applying the conditional channel $(\Phi_1, \Phi_2)$ and the resulting state will be a bipartite state from $A(K_C)^{*+} \treal A(K_{B_1} \times K_{B_2})^{*+}$. The structure of the resulting bipartite conditional state $(id \otimes (\Phi_1, \Phi_2) )(\psi)$ will not only depend on the input state $\psi$, but also on the compatibility of the channels $\Phi_1$ and $\Phi_2$. Let us assume that the channels $\Phi_1$ and $\Phi_2$ are compatible, then there is a channel $\Phi: K_A \to K_{B_1} \treal K_{B_2}$ such that $(\Phi_1, \Phi_2) = J(\Phi)$ and we have
\begin{align*}
( id \otimes (\Phi_1, \Phi_2) ) (\psi) &= (id \otimes J(\Phi) ) (\psi) \\
&= (id \otimes J') ( (id \otimes \Phi) (\psi) )
\end{align*}
where $J': A(K_{B_1} \treal K_{B_2})^{*} \to A(K_{B_1} \times K_{B_2})^{*}$, $J'(\psi) = ( (id \otimes 1)(\psi), (1 \otimes id)(\psi))$. The calculation shows that if the channels $\Phi_1$, $\Phi_2$ are compatible, then we must have
\begin{equation*}
( id \otimes (\Phi_1, \Phi_2) )(\psi) \in (id \otimes J') ( K_C \treal K_{B_1} \treal K_{B_2} )
\end{equation*}
which does not have to hold in general if the channels are not compatible. This shows that we can define steering of a state by channels as an entanglement assisted incompatibility test.

\begin{defin} \label{def:steering-steering}
The bipartite state $\psi \in A(K_C)^{*+} \treal A(K_A)^{*+}$ is steerable by channels $\Phi_1: A(K_A)^{*+} \to A(K_{B_1})^{*+}$, $\Phi_2: A(K_A)^{*+} \to A(K_{B_2})^{*+}$ if
\begin{equation*}
(id \otimes (\Phi_1, \Phi_2))(\psi) \notin (id \otimes J') ( K_C \treal K_{B_1} \treal K_{B_2} )
\end{equation*}
\end{defin}

Now we present the standard result about the connection between compatibility of the channels and steering. The result follows from our definition immediately.

\begin{coro}
The bipartite state $\psi \in A(K_C)^{*+} \treal A(K_A)^{*+}$ is not steerable by channels $\Phi_1: A(K_A)^{*+} \to A(K_B)^{*+}$, $\Phi_2: A(K_A)^{*+} \to A(K_B)^{*+}$ if the channels $\Phi_1$ and $\Phi_2$ are compatible.
\end{coro}
\begin{proof}
If the channels $\Phi_1$, $\Phi_2$ are compatible, then we have $(\Phi_1, \Phi_2) = J(\Phi)$ for some $\Phi: K_A \to K_{B_1} \treal K_{B_2}$ and for every $\psi \in K_C \treal K_A$ we have
\begin{equation*}
(id \otimes (\Phi_1, \Phi_2))(\psi) \in (id \otimes J') ( K_C \treal (K_{B_1} \treal K_{B_2}) ).
\end{equation*}
\end{proof}

\begin{prop}
The bipartite state $\psi \in A(K_C)^{*+} \treal A(K_A)^{*+}$ is not steerable by channels $\Phi_1: A(K_A)^{*+} \to A(K_B)^{*+}$, $\Phi_2: A(K_A)^{*+} \to A(K_B)^{*+}$ if $\psi \in A(K_C)^{*+} \tmin A(K_A)^{*+}$, i.e. if $\psi$ is separable.
\end{prop}
\begin{proof}
Every separable state is by definition a convex combination of product states, i.e. of states of the form $x_C \otimes x_A$, where $x_A \in K_A$, $x_C \in K_C$. Since the maps $id \otimes (\Phi_1, \Phi_2)$ and $id \otimes J'$ are linear it is sufficient to prove that for every product state $x_C \otimes x_A \in K_C \tmin K_A$ we have $(id \otimes (\Phi_1, \Phi_2))(x_C \otimes x_A) \in (id \otimes J')(K_C \treal K_{B_1} \treal K_{B_2})$. It follows by our construction in Sec. \ref{sec:prelude} that product states are not steerable by any channels as one can always take $x_C \otimes \Phi_1(x_A) \otimes \Phi_2(x_A)$. Remember that during steering, we fix not only the channels, but also the bipartite state, so the presented construction is valid.
\end{proof}

\section{Steering by measurements} \label{sec:meassteering}
We will show that the definition of steering given by Def. \ref{def:steering-steering} follows the standard definition of steering \cite{WisemanJonesDoherty-nonlocal} in the formalism introduced in \cite{Jencova-compatibility}, when we replace measurements by channels.

\begin{prop}
Let $S_1$, $S_2$ be simplexes and let $m_1: K_A \to S_1$, $m_2: K_A \to S_2$ be measurements, then a state $\psi \in K_C \treal K_A$ is steerable by $m_1$, $m_2$ if and only if
\begin{equation*}
(id \otimes (m_1, m_2))(\psi) \notin K_C \tmin (S_1 \times S_2 ).
\end{equation*}
\end{prop}
\begin{proof}
The result follows from the fact that $K_C \treal S_1 \treal S_2 = K_C \tmin S_1 \tmin S_2$.
\end{proof}

To obtain the standard definition of steering, one only needs to note that if $\xi \in K_C \tmin (S_1 \times S_2 )$, then there are $x_i \in K_C$,  $s_i \in S_1 \times S_2$ and $0 \leq \lambda_i \leq 1$ for $i \in \{1, \ldots, n\}$ such that
\begin{equation}
\xi = \sum_{i=1}^n \lambda_i x_i \otimes s_i, \label{eq:meassteering-basic}
\end{equation}
where the interpretation of $s_i$ is that it is a conditional probability, conditioned by the choice of the measurement. At this point it is straightforward to see that Eq. \eqref{eq:meassteering-basic} corresponds to \cite[Eq. (5)]{WisemanJonesDoherty-nonlocal}.

\section{Steering by quantum channels} \label{sec:quantsteering}
Steering plays an important role in quantum theory. It has found so far applications in quantum cryptography \cite{Branciard-1SDKKQD} as an intermediate step between quantum key distribution and device-independent quantum key distribution.

We will prove several results and present a simple example of steering by quantum channels. Given the standard, operational, interpretation of steering by measurements the example may seem strange, but rather expected.

Let $\Phi_1: \dens_\Ha \to \dens_\Ha$, $\Phi_2: \dens_\Ha \to \dens_\Ha$ be channels and let $|\psi^+\> = (\dH)^{-\frac{1}{2}} \sum_{i=1}^{\dH} |i i\>$ be the maximally entangled vector. We will show that the maximally entangled state $|\psi^+\>\<\psi^+|$ is steerable by the channels $\Phi_1$, $\Phi_2$ whenever they are incompatible.

The proof is rather simple as the bipartite conditional state we obtain is $(id \otimes (\Phi_1, \Phi_2) )(|\psi^+ \>\< \psi^+| )$. If the channels $\Phi_1$, $\Phi_2$ are compatible then the state $|\psi^+\>\<\psi^+|$ is not steerable by compatible channels. Now let us assume that there is a state in $\rho \in \dens_{\Ha \otimes \Ha \otimes \Ha}$ such that we have
\begin{equation}
(id \otimes (\Phi_1, \Phi_2)) (|\psi^+ \>\< \psi^+|) = (id \otimes J')(\rho), \label{eq:quantsteering-MEstate}
\end{equation}
i.e. that the state state $|\psi^+\>\<\psi^+|$ is not steerable by the channels $\Phi_1$, $\Phi_2$. Eq. \eqref{eq:quantsteering-MEstate} implies that we must have
\begin{equation*}
(id \otimes \Phi_1) (|\psi^+ \>\< \psi^+|) = \Tr_3 ( \rho ),
\end{equation*}
that, after taking trace over the second Hilbert space, gives
\begin{equation}
\dfrac{1}{\dH} \I = \Tr_{23} ( \rho ). \label{eq:quantsteering-normalization}
\end{equation}
Now the picture becomes clear: $(id \otimes \Phi_1) (|\psi^+ \>\< \psi^+|)$ is isomorphic to the Choi matrix $C(\Phi_1)$ and Eq. \eqref{eq:quantsteering-normalization} implies that the state $\rho$ must be isomorphic to a Choi matrix of some channel $\Phi$. This together with Prop. \ref{prop:quantcompat-Choi} means that Eq. \eqref{eq:quantsteering-MEstate} holds if and only if the channels are compatible. Thus we have proved:
\begin{prop} \label{prop:quantsteering-MEstate}
The maximally entangled state $|\psi^+\>\<\psi^+|$ is steerable by channels $\Phi_1: \dens_\Ha \to \dens_\Ha$, $\Phi_2: \dens_\Ha \to \dens_\Ha$ if and only if they are incompatible.
\end{prop}

We will investigate steering by unitary channels. We will see a phenomenon that is impossible to happen for steering by measurements - it is possible to steer a state when the two channels we are testing for incompatibility are two copies of the same channel. Let $U, V$ be unitary matrices, i.e. $U U^* = V V^* = \I$, where $U^*$ denotes the conjugate transpose matrix to $U$ and let $\Phi_U$, $\Phi_V$ be the corresponding unitary channels, that is for $\rho \in \dens_\Ha$ we have
\begin{align*}
\Phi_U ( \rho ) &= U \rho U^*, \\
\Phi_V ( \rho ) &= V \rho V^*.
\end{align*}
Note that we have $\Phi_\I = id$, i.e. the unitary channels given by an identity matrix is the identity channel.

\begin{prop}
The bipartite state $\rho \in \dens_{\Ha \otimes \Ha}$ is steerable by the unitary channels $\Phi_U$, $\Phi_V$ if and only if it is steerable by two copies of the identity channel $id$.
\end{prop}
\begin{proof}
The state $\rho \in \dens_{\Ha \otimes \Ha}$ is steerable by the channels $\Phi_U$, $\Phi_V$ if and only if there is a state $\sigma \in \dens_{\Ha \otimes \Ha \otimes \Ha}$ such that
\begin{align*}
\Tr_3 ( \sigma ) &= (id \otimes \Phi_U)(\rho), \\
\Tr_2 ( \sigma ) &= (id \otimes \Phi_V)(\rho).
\end{align*}
If such state $\sigma$ exists, then for $\tilde{\sigma} = (id \otimes \Phi_{U^*} \otimes \Phi_{V^*})(\sigma)$ we have
\begin{align*}
\Tr_3 ( \tilde{\sigma} ) &= \rho, \\
\Tr_2 ( \tilde{\sigma} ) &= \rho,
\end{align*}
i.e. the state $\rho$ is not steerable by two copies of $id$.

The same holds other way around by almost the same construction; if the state $\rho$ is not steerable by two copies of $id$ then it is not steerable by any unitary channels $\Phi_U$, $\Phi_V$.
\end{proof}
Note that similar result would hold if only one of the channels would be unitary, but then only that one unitary channel would be replaced by the identity map $id$.

Clearly if the state $\rho$ would be separable, then it would not be steerable by any channel. The converse does not hold, even if the state $\rho$ is entangled it still may not be steerable by any channels. We will provide a useful condition for the steerability of a given state $\rho \in \dens_{\Ha \otimes \Ha}$ that will help us to show that even if the state $\rho$ is entangled, it does not have to be steerable by any pair of channels $\Phi_1: \dens_\Ha \to \dens_\Ha$, $\Phi_2: \dens_\Ha \to \dens_\Ha$.

\begin{prop} \label{prop:quantsteering-concatenation}
The state $\rho \in \dens_{\Ha \otimes \Ha}$ is steerable by the channels $\Phi_1: \dens_\Ha \to \dens_\Ha$, $\Phi_2: \dens_\Ha \to \dens_\Ha$ only if it is steerable by two copies of the identity channel $id: \dens_\Ha \to \dens_\Ha$.
\end{prop}
\begin{proof}
Assume that the state $\rho \in \dens_{\Ha \otimes \Ha}$ is not steerable by two copies of the identity channel $id: \dens_\Ha \to \dens_\Ha$, then there exists a state $\sigma \in \dens_{\Ha \otimes \Ha \otimes \Ha}$ such that
\begin{align*}
\Tr_3 ( \sigma ) &= \rho, \\
\Tr_2 ( \sigma ) &= \rho.
\end{align*}
Let $\Phi_1: \dens_\Ha \to \dens_\Ha$, $\Phi_2: \dens_\Ha \to \dens_\Ha$ be any two channels and denote
\begin{equation*}
\tilde{\sigma} = (id \otimes \Phi_1 \otimes \Phi_2) (\sigma),
\end{equation*}
then we have
\begin{align*}
\Tr_3 ( \tilde{\sigma} ) &= (id \otimes \Phi_1) (\rho), \\
\Tr_2 ( \tilde{\sigma} ) &= (id \otimes \Phi_2) (\rho),
\end{align*}
so the state $\rho$ is not steerable by the channels $\Phi_1$, $\Phi_2$.
\end{proof}
Note that one may get other conditions for steering by replacing only one of the channels by the identity map $id$.

One may generalize this result to the general probabilistic theory but it may be rather restrictive and not as general as one would wish. One may also use the idea of the proof of Prop. \ref{prop:quantsteering-concatenation} together with the result of Prop. \ref{prop:quantsteering-MEstate} to obtain the results on compatibility of channels that are concatenations of other channels, similar to the results obtained in \cite{HeinosaariMiyadera-compOfChan}.

We will present an example of an entangled state that is not steerable by any pair of channels.
\begin{exm} \label{exm:quantsteering-W}
Let $\dH = 2$ with the standard basis $|0\>$, $|1\>$ and let $|W\> \in \Ha \otimes \Ha \otimes \Ha$ be given as
\begin{equation*}
|W\> = \dfrac{1}{\sqrt{3}} ( |001\> + |010\> + |100\>).
\end{equation*}
The projector $|W\>\<W| \in \dens_{\Ha \otimes \Ha \otimes \Ha}$ is known as W state. We have
\begin{equation*}
\rho_W = \Tr_2 ( |W\>\<W| ) = \Tr_3 ( |W\>\<W| ) \in \dens_{\Ha \otimes \Ha},
\end{equation*}
that shows that the state $\rho_W$ is not steerable by a pair of the identity channels $id: \dens_\Ha \to \dens_\Ha$, which as a result of Prop. \ref{prop:quantsteering-concatenation} means that it is not steerable by any channels $\Phi_1: \dens_\Ha \to \dens_\Ha$, $\Phi_2: \dens_\Ha \to \dens_\Ha$. Moreover it is known that the state $\rho_W$ is entangled \cite[Example 6.70]{HeinosaariZiman-MLQT}.

Since it will be useful in later calculations we will show that the state $|W\>\<W|$ is the only state from $\dens_{\Ha \otimes \Ha \otimes \Ha}$ such that $\rho_W = \Tr_2 ( |W\>\<W| ) = \Tr_3 ( |W\>\<W| )$. Let $|\varphi\> = \frac{1}{\sqrt{2}} ( |01\> + |10\>)$, then we have
\begin{equation*}
\rho_W = \dfrac{1}{3} |00\>\<00| + \dfrac{2}{3} |\varphi\>\<\varphi|.
\end{equation*}
Let $\sigma \in \dens_{\Ha \otimes \Ha \otimes \Ha}$ denote the state such that $\rho_W = \Tr_2 ( \sigma ) = \Tr_3 ( \sigma )$. We have $\rho_W |11\> = 0$ that implies $\Tr(\sigma |11\>\<11| \otimes \I ) = \Tr(\sigma |1\>\<1| \otimes \I \otimes |1\>\<1| ) = 0$ that implies $\<111|\sigma|111\> = \<110| \sigma|110\> = \<101| \sigma|101\> = 0$ as $\sigma \geq 0$. We will show that this implies $\sigma|111\> = \sigma|110\> = \sigma|101\> = 0$.

Let $A \in B_h(\Ha)$, $A \geq 0$ and let $|\psi\> \in \Ha$. Let $\Vert \psi \Vert = \sqrt{\<\psi| \psi\>}$ denote the norm given by inner product. Assume that we have $\<\psi| A |\psi\> = 0$, then
\begin{equation*}
\Vert \sqrt{A} \psi \Vert^2 = \<\sqrt{A}\psi | \sqrt{A}\psi \> = \<\psi| A |\psi\> = 0
\end{equation*}
and in conclusion we have $\sqrt{A}|\psi\> = 0$ and 
\begin{equation*}
A |\psi\> = \sqrt{A} ( \sqrt{A} |\psi\>) = 0.
\end{equation*}

Finally let us denote $|\varphi^\perp \> = \frac{1}{\sqrt{2}}( |01\> - |10\> )$. We have $\rho_W |\varphi^\perp\> = 0$ that implies $\Tr( \sigma |\varphi^\perp\>\<\varphi^\perp| \otimes \I) = 0$ which yields $\sigma |\varphi^\perp 0 \> = \sigma |\varphi^\perp 1 \> = 0$. We still use the shorthand $|\varphi^\perp 0 \> = |\varphi^\perp\> \otimes |0\>$.

The eight vectors $|000\>$, $|001\>$, $|\varphi 0\>$, $|\varphi 1\>$, $|\varphi^\perp 0\>$, $|\varphi^\perp 1\>$, $|110\>$, $|111\>$ form an orthonormal basis of $\Ha \otimes \Ha \otimes \Ha$. We have already showed that we must have
\begin{equation*}
\sigma |\varphi^\perp 0 \> = \sigma |\varphi^\perp 1 \> = \sigma |110\> = \sigma | 111 \> = 0
\end{equation*}
so in general we must have
\begin{align*}
\sigma &= a_{00} |000\>\<000| + a_{01} |001\>\<001| + a_{\varphi0} |\varphi 0\>\< \varphi 0| \\
&+ a_{\varphi1} |\varphi1\>\<\varphi1| + b_1 |000\>\<001| + \bar{b}_1 |001\>\<000| \\
&+ b_2 |000\>\< \varphi 0| + \bar{b}_2 |\varphi 0\>\<000| + b_3 |000\>\<\varphi 1| \\
&+ \bar{b}_3 | \varphi 1\>\< 000| + b_4 |001\>\<\varphi 0| + \bar{b}_4 |\varphi 0\>\<001| \\
&+ b_5 |001\>\<\varphi 1| + \bar{b}_5 | \varphi 1\>\<001| + b_6 |\varphi 0 \>\< \varphi 1| \\
&+ \bar{b}_6 |\varphi 1\>\< \varphi 0|.
\end{align*}

Using the above expression for $\sigma$ we get
\begin{align*}
\Tr_2 (\sigma) &= a_{00} |00\>\<00| + a_{01} |01\>\<01| + \dfrac{a_{\varphi 0}}{2} \I \otimes |0\>\<0| \\
&+ \dfrac{a_{\varphi 1}}{2} \I \otimes |1 \>\<1| + b_1 |00\>\<01| + \bar{b}_1 |01\>\<00| \\
&+ \dfrac{b_2}{\sqrt{2}} |00\>\<10| + \dfrac{\bar{b}2}{\sqrt{2}} |10\>\<00| + \dfrac{b_3}{\sqrt{2}} |00\>\<11| \\
&+ \dfrac{\bar{b}_3}{\sqrt{2}} |11\>\<00| + \dfrac{b_4}{\sqrt{2}} |01\>\<10| + \dfrac{\bar{b}_4}{\sqrt{2}} |10\>\<01| \\
&+ \dfrac{b_5}{\sqrt{2}} |01\>\<11| + \dfrac{\bar{b}_5}{\sqrt{2}} |11\>\<01| + \dfrac{b_6}{2} \I \otimes |0\>\<1| \\
&+ \dfrac{\bar{b}_6}{2} \I \otimes |1\>\<0|,
\end{align*}
that implies $a_{0 0} = a_{\varphi 1} = 0$, $a_{\varphi 0} = \frac{2}{3}$, $a_{01} = \frac{1}{3}$, $b_1 = b_2 = b_3 = b_5 = b_6 = 0$ and $b_4 = \frac{\sqrt{2}}{3}$. In conclusion we have
\begin{align*}
\sigma =& \dfrac{1}{3} ( |001\>\<001| + 2 |\varphi 0\>\<\varphi 0| \\
&+ \sqrt{2} |001\>\<\varphi 0| + \sqrt{2} |\varphi 0\>\< 001| ) \\
=& |W\>\<W|.
\end{align*}
\end{exm}

\section{Bell nonlocality} \label{sec:nonlocal}
Bell nonlocality is, similarly to steering, a phenomenom that we do not find in classical theory, but is often used in quantum theory. Bell nonlocality \cite{Bell-ineq} was formulated as a response to the well-known EPR paradox \cite{EinsteinPodolskyRosen-paradox}. Although in the original formulation, the operational idea was different than the one we will present, we will see that Bell nonlocality may be understood as an incompatibility test, in the same way as steering.

Let us assume that we have four parties: Alice, Bob, Charlie and Dan. Alice has two channels $\Phi_1^A: K_A \to K_{B_1}$ and $\Phi_2^A: K_A \to K_{B_2}$ that she can use to send a state to Bob and Charlie has two channels $\Phi_1^C: K_C \to K_{D_1}$ and $\Phi_2^C: K_C \to K_{D_2}$ that he can use to send a state to Dan. Assume that $K_C \treal K_A$ is defined and let $\psi \in K_C \treal K_A$ be a bipartite state shared by Alice and Charlie. The idea that we use to define Bell nonlocality is very simple: if we were able to use $(id \otimes (\Phi_1^A, \Phi_2^A))(\psi)$ and $((\Phi_1^C, \Phi_2^C) \otimes id)(\psi)$ as non-trivial incompatibility test, we may as well investigate whether $((\Phi_1^C, \Phi_2^C) \otimes (\Phi_1^A, \Phi_2^A))(\psi)$ provides an incompatibility test in the same manner.

\begin{defin}
Let $\psi \in K_C \treal K_A$ and let
\begin{align*}
\Phi_1^A &: K_A \to K_{B_1}, \\
\Phi_2^A &: K_A \to K_{B_2}, \\
\Phi_1^C &: K_C \to K_{D_1}, \\
\Phi_2^C &: K_C \to K_{D_2}
\end{align*}
be channels. We call the state $((\Phi_1^C, \Phi_2^C) \otimes (\Phi_1^A, \Phi_2^A))(\psi)$ bipartite biconditional state.
\end{defin}

Assume that the channels $\Phi_1^A$ and $\Phi_2^A$ are compatible, so that we have $(\Phi_1^A, \Phi_2^A) = J(\Phi^A)$ for some channel $\Phi^A: K_A \to K_{B_1} \treal K_{B_2}$ and also that the channels $\Phi_1^C$ and $\Phi_2^C$ are compatible, so there is a channel $\Phi^C: K_C \to K_{D_1} \treal K_{D_2}$ such that $(\Phi_1^A, \Phi_2^A) = J(\Phi^A)$. Let $\psi \in K_C \treal K_A$, then we have
\begin{equation*}
((\Phi_1^C, \Phi_2^C) \otimes (\Phi_1^A, \Phi_2^A))(\psi) = (J' \otimes J') ( (\Phi_C \otimes \Phi_A) (\psi) ),
\end{equation*}
where the maps $J'$ are defined as before, with the exception that we denote them the same even though they map different spaces.

We present a definition of Bell nonlocality using the same line of thinking as we used in Def. \ref{def:steering-steering}. For simplicity we will denote
\begin{equation*}
Q_{DC} = (J' \otimes J') ( K_{D_1} \treal K_{D_2} \treal K_{C_1} \treal K_{C_2} ).
\end{equation*}
\begin{defin} \label{def:nonlocal-nonlocal}
Let $\psi \in K_C \treal K_A$ be a bipartite state and let $\Phi_1^A: K_A \to K_{B_1}$, $\Phi_2^A: K_A \to K_{B_2}$, $\Phi_1^C: K_A \to K_{C_1}$ and $\Phi_2^C: K_A \to K_{D_2}$ be channels. We say that the bipartite biconditional state $((\Phi_1^C, \Phi_2^C) \otimes (\Phi_1^A, \Phi_2^A))(\psi)$ is Bell nonlocal if
\begin{equation*}
((\Phi_1^C, \Phi_2^C) \otimes (\Phi_1^A, \Phi_2^A))(\psi) \notin Q_{DC}.
\end{equation*}
Otherwise we call the bipartite biconditional state Bell local.
\end{defin}

The following result follows immediatelly from Def. \ref{def:nonlocal-nonlocal}.
\begin{coro}
Let $\psi \in K_C \treal K_A$ be a bipartite state and let $\Phi_1^A: K_A \to K_{B_1}$, $\Phi_2^A: K_A \to K_{B_2}$, $\Phi_1^C: K_A \to K_{C_1}$ and $\Phi_2^C: K_A \to K_{D_2}$ be channels. The bipartite biconditional state $((\Phi_1^C, \Phi_2^C) \otimes (\Phi_1^A, \Phi_2^A))(\psi)$ is Bell nonlocal only if the channels $\Phi_1^A$, $\Phi_2^A$ and $\Phi_1^C$, $\Phi_2^C$ are incompatible.
\end{coro}

We will show that entanglement plays a key role in Bell nonlocality.
\begin{prop}
Let $\psi \in K_C \tmin K_A$ be a separable bipartite state and let $\Phi_1^A: K_A \to K_{B_1}$, $\Phi_2^A: K_A \to K_{B_2}$, $\Phi_1^C: K_A \to K_{C_1}$ and $\Phi_2^C: K_A \to K_{D_2}$ be channels. The bipartite biconditional state $((\Phi_1^C, \Phi_2^C) \otimes (\Phi_1^A, \Phi_2^A))(\psi)$ is Bell local.
\end{prop}
\begin{proof}
It is again sufficient to consider $\psi = x_C \otimes x_A$ for $x_A \in K_A$, $x_C \in K_C$ due to the linearity of the maps $(\Phi_1^A, \Phi_2^A)$ and $(\Phi_1^C, \Phi_2^C)$. Consider the state $\varphi \in K_{D_1} \treal K_{D_2} \treal K_{C_1} \treal K_{C_2}$ given as
\begin{equation*}
\varphi = \Phi_1^C(x_C) \otimes \Phi_2^C(x_C) \otimes \Phi_1^A(x_A) \otimes \Phi_2^A(x_A),
\end{equation*}
then we have
\begin{equation*}
((\Phi_1^C, \Phi_2^C) \otimes (\Phi_1^A, \Phi_2^A))(\psi) = (J' \otimes J') (\varphi).
\end{equation*}
\end{proof}

\section{Bell nonlocality of measurements} \label{sec:measnonlocal}
We will again show that the Def. \ref{def:nonlocal-nonlocal} follows the standard definition of Bell nonlocality \cite{WisemanJonesDoherty-nonlocal} in the formalism of \cite{Jencova-compatibility}.

\begin{prop}
Let $S_1^A$, $S_2^A$, $S_1^C$ and $S_2^C$, be simplexes and let $m_1^A: K_A \to S_1^A$, $m_2^A: K_A \to S_2^A$, $m_1^C: K_C \to S_1^C$, $m_2^C: K_C \to S_2^C$ be measurements. Let $\psi \in K_C \treal K_A$, then the bipartite biconditional state $((m_1^C, m_2^C) \otimes (m_1^A, m_2^A))(\psi)$ is Bell nonlocal if
\begin{equation*}
((m_1^C, m_2^C) \otimes (m_1^A, m_2^A))(\psi) \notin (S_1^C \times S_2^C) \tmin (S_1^A \times S_2^A).
\end{equation*}
\end{prop}
\begin{proof}
By direct calculation we have
\begin{align*}
Q_{CD} &= (J' \otimes J') ( S_1^C \tmin S_2^C \tmin S_1^A \tmin S_2^A ) \\
&= ( S_1^C \times S_2^C) \tmin ( S_1^A \times S_2^A).
\end{align*}
\end{proof}

One may again use the interpretation that both $S_1^C \times S_2^C$ and $S_1^A \times S_2^A$ are spaces of conditional measurement probabilities, so if we have $\psi \in ( S_1^C \times S_2^C) \tmin ( S_1^A \times S_2^A)$ then we must have $0 \leq \lambda_i \leq 1$, for $i \in \{1, \ldots, n\}$, $\sum_{i=1}^n \lambda_i = 1$, such that
\begin{equation*}
\psi = \sum_{i=1}^n \lambda_i s_i^C \otimes s_i^A
\end{equation*}
where in standard formulations both $s_i^C \in S_1^C \times S_2^C$ and $s_i^A \in _1^A \times S_2^A$ are represented by probabilities, i.e. by numbers, so the tensor product between them is omitted.

We will provide proof of the standard and well-known result about connection of steering and Bell nonlocality of measurements.
\begin{prop} \label{prop:measnonlocal-steering4nonlocal}
Let $S_1^A$, $S_2^A$, $S_1^C$ and $S_2^C$, be simplexes and let $m_1^A: K_A \to S_1^A$, $m_2^A: K_A \to S_2^A$, $m_1^C: K_C \to S_1^C$, $m_2^C: K_C \to S_2^C$ be measurements. Let $\psi \in K_C \treal K_A$. If
\begin{equation*}
(id \otimes (m_1^A, m_2^A))(\psi) \in K_C \tmin (S_1^A \times S_2^A),
\end{equation*}
i.e. if the bipartite state is not steerable by measurements $m_1^A$, $m_2^A$, then
\begin{equation*}
((m_1^C, m_2^C) \otimes (m_1^A, m_2^A))(\psi) \in (S_1^C \times S_2^C) \tmin (S_1^A \times S_2^A).
\end{equation*}
\end{prop}
\begin{proof}
Let
\begin{equation*}
(id \otimes (m_1^A, m_2^A))(\psi) \in K_C \tmin (S_1^A \times S_2^A),
\end{equation*}
then for $n \in \mathbb{N}$, $i \in \{1, \ldots, n\}$, there are $0 \leq \lambda_i \leq 1$, $x_i \in K_C$ and $s_i \in S_1^A \times S_2^A$, $\sum_{i=1}^n \lambda_i = 1$, such that we have
\begin{equation*}
(id \otimes (m_1^A, m_2^A))(\psi) = \sum_{i=1}^n \lambda_i x_i \otimes s_i.
\end{equation*}
We get
\begin{equation*}
((m_1^C, m_2^C) \otimes (m_1^A, m_2^A))(\psi) = \sum_{i=1}^n \lambda_i (m_1^C, m_2^C)(x_i) \otimes s_i
\end{equation*}
and since we have $(m_1^C, m_2^C)(x_i) = (m_1^C(x_i), m_2^C(x_i)) \in S_1^C \times S_2^C$ we have
\begin{equation*}
((m_1^C, m_2^C) \otimes (m_1^A, m_2^A))(\psi) \in (S_1^C \times S_2^C) \tmin (S_1^A \times S_2^A).
\end{equation*}
\end{proof}
Note that the same result would also hold for steering by the measurements $m_1^C$, $m_2^C$.

One may think that steering is somehow half of Bell nonlocality, or that it is some middle step towards Bell nonlocality as even our constructions in Sec. \ref{sec:steering} and \ref{sec:nonlocal} would point to such a result. We will show that this is not true in general, as we will provide a counter-example using quantum channels in example \ref{exm:quantnonlocal-nonlocalnotsteer}.

\section{Bell nonlocality of quantum channels} \label{sec:quantnonlocal}
Bell nonlocality of quantum measurements is a deeply studied topic in quantum theory, with several applications in various device-independent protocols \cite{Ekert-crypto, AcinBrunnerGisinMassarPironio-DIQKD, BruknerZukowskiPanZeilinger-BellComCompl, HwangSuBae-MDIQKD}, randomness generation and randomness expansion \cite{ColbeckKent-random, BischofKampermannBruss-random} and others, for a recent review on Bell nonlocality see \cite{BrunnerCavalcantiPironioScarani-nonlocal}.

Bell nonlocality of quantum channels follows very similar rules to steering by quantum channels. We will derive results specific for quantum theory in the same manners as in Sec. \ref{sec:quantsteering}.

\begin{prop} \label{prop:quantnonlocal-necessary}
Let $\rho \in \dens_{\Ha \otimes \Ha}$ and let $\Phi_1^1: \dens_\Ha \to \dens_\Ha$, $\Phi_2^1: \dens_\Ha \to \dens_\Ha$, $\Phi_1^2: \dens_\Ha \to \dens_\Ha$, $\Phi_2^2: \dens_\Ha \to \dens_\Ha$ be channels. The bipartite biconditional state $((\Phi_1^1, \Phi_2^1) \otimes (\Phi_1^2, \Phi_2^2))(\rho)$ is Bell nonlocal only if the bipartite biconditional state $((id, id) \otimes (id, id))(\rho)$ is Bell nonlocal.
\end{prop}
\begin{proof}
If the bipartite biconditional state $((id, id) \otimes (id, id))(\rho)$ is Bell local, then there exist $\sigma \in \dens_{\Ha \otimes \Ha \otimes \Ha \otimes \Ha}$ such that
\begin{align*}
\Tr_{24} (\sigma) &= \rho, \\
\Tr_{23} (\sigma) &= \rho, \\
\Tr_{14} (\sigma) &= \rho, \\
\Tr_{13} (\sigma) &= \rho.
\end{align*}
Let
\begin{equation*}
\tilde{\sigma} = (\Phi_1^1 \otimes \Phi_2^1 \otimes \Phi_1^2 \otimes \Phi_2^2) (\sigma)
\end{equation*}
then
\begin{align*}
\Tr_{24} (\tilde{\sigma}) &= (\Phi_1^1 \otimes \Phi_1^2 ) (\rho), \\
\Tr_{23} (\tilde{\sigma}) &= (\Phi_1^1 \otimes \Phi_2^2 ) (\rho), \\
\Tr_{14} (\tilde{\sigma}) &= (\Phi_2^1 \otimes \Phi_1^2 ) (\rho), \\
\Tr_{13} (\tilde{\sigma}) &= (\Phi_2^1 \otimes \Phi_2^2 ) (\rho).
\end{align*}
\end{proof}
Note that again we do not have to replace all of the channels by the identity channels $id$, but we may replace only some.

\begin{prop} \label{prop:quantnonlocal-unitaryIff}
Let $\rho \in \dens_{\Ha \otimes \Ha}$ and let $\Phi_1^1: \dens_\Ha \to \dens_\Ha$, $\Phi_2^1: \dens_\Ha \to \dens_\Ha$, $\Phi_1^2: \dens_\Ha \to \dens_\Ha$, $\Phi_2^2: \dens_\Ha \to \dens_\Ha$ be channels, moreover let $\Phi_1^1 = \Phi_U$ be a unitary channel given by the unitary matrix $U$, then the bipartite biconditional state $((\Phi_U, \Phi_2^1) \otimes (\Phi_1^2, \Phi_2^2))(\rho)$ is Bell nonlocal if and only if the bipartite biconditional state $((id, \Phi_2^1) \otimes (\Phi_1^2, \Phi_2^2))(\rho)$ is Bell nonlocal.
\end{prop}
\begin{proof}
Using the very same idea as before, if the  bipartite biconditional state $((\Phi_U, \Phi_2^1) \otimes (\Phi_1^2, \Phi_2^2))(\rho)$ is Bell local, then there is $\sigma \in \dens_{\Ha \otimes \Ha \otimes \Ha \otimes \Ha}$ such that
\begin{align*}
\Tr_{24} (\sigma) &= (\Phi_U \otimes \Phi_1^2 ) (\rho), \\
\Tr_{23} (\sigma) &= (\Phi_U \otimes \Phi_2^2 ) (\rho), \\
\Tr_{14} (\sigma) &= (\Phi_2^1 \otimes \Phi_1^2 ) (\rho), \\
\Tr_{13} (\sigma) &= (\Phi_2^1 \otimes \Phi_2^2 ) (\rho).
\end{align*}
Let
\begin{equation*}
\tilde{\sigma} = (\Phi_{U^*} \otimes id \otimes id \otimes id) (\sigma)
\end{equation*}
then we get
\begin{align*}
\Tr_{24} (\tilde{\sigma}) &= (id \otimes \Phi_1^2 ) (\rho), \\
\Tr_{23} (\tilde{\sigma}) &= (id \otimes \Phi_2^2 ) (\rho), \\
\Tr_{14} (\tilde{\sigma}) &= (\Phi_2^1 \otimes \Phi_1^2 ) (\rho), \\
\Tr_{13} (\tilde{\sigma}) &= (\Phi_2^1 \otimes \Phi_2^2 ) (\rho).
\end{align*}
\end{proof}
One may obtain similar results if some other of the channels $\Phi_1^1, \Phi_2^1, \Phi_1^2, \Phi_2^2$ is unitary as well as if more or even all of them are unitary.

The most iconic and most studied aspect of Bell nonlocality are the Bell inequalities. We are going to present a version of CHSH inequality for quantum channels. Assume that $\dH = 2$ and let $|0\>, |1\>$ denote any orthonormal basis of $\Ha$. We will use the shorthand $|00\> = |0\> \otimes |0\>$. Let $i, j \in \{1, 2\}$ and let
\begin{align*}
E(\Phi_i^1, \Phi_j^2) &= \<00| (\Phi_i^1 \otimes \Phi_j^2 )(\rho) |00\> \\
&- \<01| (\Phi_i^1 \otimes \Phi_j^2 )(\rho) |01\> \\
&- \<10| (\Phi_i^1 \otimes \Phi_j^2 )(\rho) |10\> \\
&+ \<11| (\Phi_i^1 \otimes \Phi_j^2 )(\rho) |11\> \\
&= \Tr ( (\Phi_i^1 \otimes \Phi_j^2 )(\rho) A )
\end{align*}
where
\begin{equation*}
A = |00\>\<00| - |01\>\<01| - |10\>\<10| + |11\>\<11|.
\end{equation*}
The quantity $E(\Phi_i^1, \Phi_j^2)$ is to be interpreted as the correlation between the marginals $\Tr_1((\Phi_i^1 \otimes \Phi_j^2 )(\rho))$ and $\Tr_2((\Phi_i^1 \otimes \Phi_j^2 )(\rho))$. Since we have $-\I \leq A \leq \I$ it is straightforward that we have $-1 \leq E(\Phi_i^1, \Phi_j^2) \leq 1$. Define a quantity
\begin{equation*}
X_\rho = E(\Phi_1^1, \Phi_1^2) + E(\Phi_1^1, \Phi_2^2) + E(\Phi_2^1, \Phi_1^2) - E(\Phi_2^1, \Phi_2^2),
\end{equation*}
we will show that $X_\rho$ corresponds to the quantity used in CHSH inequality. It is straightforward to see that $-4 \leq X_\rho \leq 4$ is the algebraic bound on $X_\rho$.

\begin{prop} \label{prop:quantnonlocal-ineq}
If the biconditional bipartite state $((\Phi_1^1, \Phi_2^1) \otimes (\Phi_1^2, \Phi_2^2) )(\rho)$ is Bell local, then we have $-2 \leq X_\rho \leq 2$.
\end{prop}
\begin{proof}
If the biconditional bipartite state $((\Phi_1^1, \Phi_2^1) \otimes (\Phi_1^2, \Phi_2^2) )(\rho)$ is Bell local then there is $\sigma \in \dens_{\Ha \otimes \Ha \otimes \Ha \otimes \Ha}$ such that
\begin{align*}
\Tr_{24} (\sigma) &= (\Phi_1^1 \otimes \Phi_1^2 ) (\rho), \\
\Tr_{23} (\sigma) &= (\Phi_1^1 \otimes \Phi_2^2 ) (\rho), \\
\Tr_{14} (\sigma) &= (\Phi_2^1 \otimes \Phi_1^2 ) (\rho), \\
\Tr_{13} (\sigma) &= (\Phi_2^1 \otimes \Phi_2^2 ) (\rho).
\end{align*}
This yields
\begin{align*}
E(\Phi_1^1, \Phi_1^2) &= \Tr ( (\Phi_1^1 \otimes \Phi_1^2 ) (\rho) A ) =  \Tr ( \Tr_{24} (\sigma) A ) \\
&= \Tr ( \sigma ( |0\>\<0| \otimes \I \otimes |0\>\<0| \otimes \I \\
&- |0\>\<0| \otimes \I \otimes |1\>\<1| \otimes \I \\
&- |1\>\<1| \otimes \I \otimes |0\>\<0| \otimes \I \\
&+ |1\>\<1| \otimes \I \otimes |1\>\<1| \otimes \I ) ).
\end{align*}
In the same manner we get
\begin{align*}
E(\Phi_1^1, \Phi_2^2) = \Tr ( \sigma ( &|0\>\<0| \otimes \I \otimes \I \otimes |0\>\<0| \\
- &|0\>\<0| \otimes \I \otimes \I \otimes |1\>\<1| \\
- &|1\>\<1| \otimes \I \otimes \I \otimes |0\>\<0| \\
+ &|1\>\<1| \otimes \I \otimes \I \otimes |1\>\<1| ) ),
\end{align*}
\begin{align*}
E(\Phi_2^1, \Phi_1^2) = \Tr ( \sigma (  &\I \otimes |0\>\<0| \otimes |0\>\<0| \otimes \I \\
- &\I \otimes |0\>\<0| \otimes |1\>\<1| \otimes \I \\
- &\I \otimes |1\>\<1| \otimes |0\>\<0| \otimes \I \\
+ &\I \otimes |1\>\<1| \otimes |1\>\<1| \otimes \I ) )
\end{align*}
and
\begin{align*}
E(\Phi_2^1, \Phi_2^2) = \Tr ( \sigma ( &\I \otimes |0\>\<0| \otimes \I \otimes |0\>\<0| \\
- &\I \otimes |0\>\<0| \otimes \I \otimes |1\>\<1| \\
- &\I \otimes |1\>\<1| \otimes \I \otimes |0\>\<0| \\
+ &\I \otimes |1\>\<1| \otimes \I \otimes |1\>\<1| ) ).
\end{align*}
Together we get
\begin{align*}
X_\rho = 2 \Tr( \sigma (& |0000\>\<0000| + |0001\>\<0001| \\
- &|0010\>\<0010| - |0011\>\<0011| \\
+ &|0100\>\<0100| - |0101\>\<0101| \\
+ &|0110\>\<0110| - |0111\>\<0111| \\
- &|1000\>\<1000| + |1001\>\<1001| \\
- &|1010\>\<1010| + |1011\>\<1011| \\
- &|1100\>\<1100| - |1101\>\<1101| \\
+ &|1110\>\<1110| + |1111\>\<1111| ) )
\end{align*}
that implies $-2 \leq X_\rho \leq 2$.
\end{proof}

At this point one may ask whether there exists an equivalent of Tsirelson bound \cite{Cirelson-bound} for the inequality given by Prop. \ref{prop:quantnonlocal-ineq}, or what is the maximum violation of the aforementioned inequality. We will show that the Tsirelson bound $2\sqrt{2}$ is both reachable and maximum violation by quantum channels.
\begin{prop}
For any state $\rho \in \dens_{\Ha \otimes \Ha}$ and any four channels $\Phi_1^1: \dens_\Ha \to \dens_\Ha$, $\Phi_2^1: \dens_\Ha \to \dens_\Ha$, $\Phi_1^2: \dens_\Ha \to \dens_\Ha$, $\Phi_2^2: \dens_\Ha \to \dens_\Ha$ we have
\begin{equation*}
X_\rho \leq 2 \sqrt{2}.
\end{equation*}
\end{prop}
\begin{proof}
We define the adjoint channel $\Phi_1^{1*}$ to channel $\Phi_1^1$ as the the linear map $\Phi_1^{1*}: B_h(\Ha) \to B_h(\Ha)$ such that for all $\sigma \in \dens_\Ha$ and $E \in B_h(\Ha)$, $0 \leq E \leq \I$ we have
\begin{equation*}
\Tr( \Phi_1^1 (\sigma) E) = \Tr( \sigma \Phi_1^{1*}(E) ).
\end{equation*}
Since $\Phi_1^1$ is a channel we have $0 \leq \Phi_1^{1*}(E) \leq \I$ and $\Phi_1^{1*}(\I) = \I$. This approach of mapping effects instead of states is called the Heisenberg picture.

Let $i, j \in \{1, 2\}$, then we have
\begin{equation*}
\Tr( (\Phi^1_i \otimes \Phi^2_j)(\rho) |00\>\<00|) = \Tr( \rho \Phi_i^{1*}(|0\>\<0|) \otimes \Phi_j^{2*}(|0\>\<0|)).
\end{equation*}
Denoting
\begin{align*}
M_i^1 &= \Phi_i^{1*}(|0\>\<0|)  \\
M_j^2 &= \Phi_j^{1*}(|0\>\<0|)  \\
\end{align*}
we see that we have
\begin{align*}
E(\Phi_i^1, \Phi_j^2) &= \Tr( \rho M_i^1 \otimes M_j^2) - \Tr( \rho (\I-M_i^1) \otimes M_j^2) \\
&- \Tr( \rho M_i^1 \otimes (\I-M_j^2)) \\
&+ \Tr( \rho (\I-M_i^1) \otimes (\I-M_j^2)) \\
&= E(M_i^1, M_j^2),
\end{align*}
where $E(M_i^1, M_j^2)$ is a correlation for the two-outcome measurements given by the effects $M_i^1$ and $M_j^2$. It is well known result \cite{Cirelson-bound} that we always have
\begin{align*}
E(M_1^1, M_1^2) + E(M_1^1, M_2^2)& \\
+ E(M_2^1, M_1^2) - E(M_2^1, M_2^2)& \leq 2 \sqrt{2}.
\end{align*}
\end{proof}

It is very intuitive that the Tsirelson bound, reachable by measurements, will be also reachable by channels. To prove this, let $M, N \in B_h(\Ha)$, $0 \leq M \leq \I$, $0 \leq N \leq \I$ and define channels $\Phi_M: B_h(\Ha) \to B_h(\Ha)$, $\Phi_N: B_h(\Ha) \to B_h(\Ha)$ such that for $\sigma \in \dens_\Ha$ we have
\begin{align*}
\Phi_M(\sigma) &= \Tr(\sigma M) |0\>\<0| +  \Tr(\sigma (\I-M)) |1\>\<1|, \\
\Phi_N(\sigma) &= \Tr(\sigma N) |0\>\<0| +  \Tr(\sigma (\I-N)) |1\>\<1|.
\end{align*}
It is easy to verify that the maps $\Phi_M$, $\Phi_N$ are quantum channels and that they are also measurements as they map the state space $\dens_\Ha$ to the simplex $\conv\{ |0\>\<0|, |1\>\<1| \}$. Let $\rho \in \dens_{\Ha \otimes \Ha}$, then we have
\begin{align*}
\Tr( (\Phi_M \otimes \Phi_N) (\rho) A ) &= \Tr( \rho (M \otimes N - (\I - M) \otimes N) \\
&- \Tr( \rho (M \otimes (\I - N) ) ) \\
&+ \Tr( \rho ((\I - M) \otimes (\I - N) ) ) \\
&= E(M, N).
\end{align*}
This proves that any set of correlations and any violation of CHSH inequality reachable by measurements is also reachable by quantum channels as a violation of the bound given by Prop. \ref{prop:quantnonlocal-ineq}.

To generalize the proposed inequality one may replace the projectors $|0\>\<0|$ and $|1\>\<1|$ by any pair of effects $M, N \in B_h(\Ha)$, $0 \leq M \leq \I$, $0 \leq N \leq \I$ and have
\begin{equation*}
A = M \otimes N - (\I - M) \otimes N - M \otimes (\I - N) + (\I - M) \otimes (\I - N).
\end{equation*}

From now on we will consider a special case. Keep $\dH = 2$ and let
\begin{equation*}
|\psi^+\>\<\psi^+| = \dfrac{1}{2} ( |00\>\<00| + |11\>\<00| + |00\>\<11| + |11\>\<11| )
\end{equation*}
be the maximally entangled state, let $U_1$, $U_2$, $V_1$, $V_2$ be unitary matrices and let $\Phi_1^1 = \Phi_{U_1}$, $\Phi_2^1 = \Phi_{U_2}$, $\Phi_1^2 = \Phi_{V_1}$, $\Phi_2^2 = \Phi_{V_2}$ be unitary channels given by the respective unitary matrices. We will consider the bipartite biconditional state $((\Phi_{U_1}, \Phi_{U_2}) \otimes (\Phi_{V_1}, \Phi_{V_2}))(|\psi^+\>\<\psi^+|)$ and we will show that the correlations for the given bipartite biconditional state are of a particular nice form. We have
\begin{equation*}
(\Phi_{U_i} \otimes \Phi_{V_j})(|\psi^+\>\<\psi^+|) = (id \otimes \Phi_{V_j U_i^T})(|\psi^+\>\<\psi^+|)
\end{equation*}
where $i, j \in \{1, 2\}$ and for $U^T$ denotes the transpose of the matrix $U$. For the correlation we have
\begin{align}
E(\Phi_{U_i}, \Phi_{V_j}) &= \Tr( (id \otimes \Phi_{V_j U_i^T})(|\psi^+\>\<\psi^+|) A ) \nonumber \\
&= \dfrac{1}{2} ( |\<0| V_j U_i^T |0\>|^2 + |\<1| V_j U_i^T |1\>|^2 \nonumber \\
&- |\<0| V_j U_i^T |1\>|^2 - |\<1| V_j U_i^T |0\>|^2 ). \label{eq:quantnonlocal-unitaryMEstate}
\end{align}

\begin{figure}
\includegraphics[width=\linewidth]{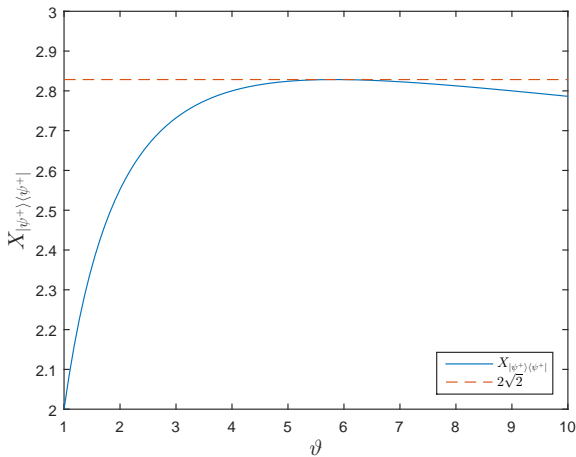}
\caption{The blue solid line is $X_{|\psi^+\>\<\psi^+|}$ as a function of the parameter $\vartheta \in [1, 10]$ when we consider the bipartite biconditional state $(\Phi_{U_1}, \Phi_{U_2}) \otimes (\Phi_{V_1}, \Phi_{V_2})(|\psi^+\>\<\psi^+|)$ from example \ref{exm:quantnonlocal-unit-num}. The red dashed line corresponds to the Tsirelson bound $2 \sqrt{2}$.} \label{fig:quantnonlocal-unitary}
\end{figure}

We will provide an example of a violation of the bound given by Prop. \ref{prop:quantnonlocal-ineq} by incompatible unitary channels.
\begin{exm} \label{exm:quantnonlocal-unit-num}
Let $\dH = 2$ and let $\vartheta \in \mathbb{R}$ be a parameter. Let $U_1$, $U_2$, $V_1$, $V_2$ be unitary matrices given as
\begin{align*}
U_1 &= \dfrac{1}{\sqrt{2}}
\begin{pmatrix}
1 && 1 \\
1 && -1
\end{pmatrix}, \\
U_2 &=
\begin{pmatrix}
1 && 0 \\
0 && 1
\end{pmatrix}, \\
V_1 &= \dfrac{1}{\sqrt{1 + \vartheta}}
\begin{pmatrix}
\sqrt{\vartheta} && 1 \\
1 && -\sqrt{\vartheta}
\end{pmatrix}, \\
V_2 &= \dfrac{1}{\sqrt{1 + \vartheta}}
\begin{pmatrix}
1 && \sqrt{\vartheta} \\
\sqrt{\vartheta} && -1
\end{pmatrix}.
\end{align*}
Consider the bipartite biconditional state $(\Phi_{U_1}, \Phi_{U_2}) \otimes (\Phi_{V_1}, \Phi_{V_2})(|\psi^+\>\<\psi^+|)$. Using Eq. \eqref{eq:quantnonlocal-unitaryMEstate} we can obtain $X_{|\psi^+\>\<\psi^+|}$ as a function of $\vartheta$. The function is plotted in Fig. \ref{fig:quantnonlocal-unitary}, where it is shown that for certain values of $\vartheta$ the bipartite biconditional state violates the bound given by Prop. \ref{prop:quantnonlocal-ineq}.

It is also easy to see that the bipartite biconditional state $((id, id) \otimes (id, id))(|\psi^+\>\<\psi^+|)$ does not violate the bound given by Prop. \ref{prop:quantnonlocal-ineq}, because all of the correlations are the same, yet according to Prop. \ref{prop:quantnonlocal-unitaryIff} we know that it must be a Bell nonlocal bipartite biconditional state. This shows that not all Bell nonlocal bipartite biconditional states violate the inequality given by Prop. \ref{prop:quantnonlocal-ineq}.
\end{exm}

One may wonder whether there is or is not a connection between steering and Bell nonlocality. As we have already showed in Prop. \ref{prop:measnonlocal-steering4nonlocal}, for measurements Bell nonlocality implies steering. We will show that for channels the same does not hold.
\begin{exm} \label{exm:quantnonlocal-nonlocalnotsteer}
Let $\dH = 2$. Let $\rho_W \in \dens_{\Ha \otimes \Ha}$ be given as in example \ref{exm:quantsteering-W} as a partial trace over the state $|W\>\<W|$. We already know that the state $\rho_W$ is not steerable by any pair of channels. Consider the bipartite biconditional state $((id, id) \otimes (id, id))(\rho_W)$, if it is Bell local, then there must be a state $\sigma \in \dens_{\Ha \otimes \Ha \otimes \Ha \otimes \Ha}$ such that
\begin{equation*}
\Tr_{13}(\sigma) = \Tr_{14} (\sigma) = \Tr_{23}(\sigma) = \Tr_{24} (\sigma) = \rho_W.
\end{equation*}
Observe that $\Tr_1(\sigma) \in \dens_{\Ha \otimes \Ha \otimes \Ha}$ is such that $\Tr_3 (\Tr_1(\sigma)) = \Tr_4 (\Tr_1(\sigma)) = \rho_W$ which implies that, according to our calculations in example \ref{exm:quantsteering-W}, we must have
\begin{equation*}
\Tr_1(\sigma) = |W\>\<W|.
\end{equation*}
According to \cite[Lemma 3]{BarnumBarretLeiferWilce-noBroadcast} this implies that there is a state $\rho \in \dens_\Ha$ such that $\sigma = \rho \otimes |W\>\<W|$. This implies that we have $\Tr_{23} (\sigma) = \rho \otimes \dfrac{1}{3} (2 |0\>\<0| + |1\>\<1|)$ which is clearly a separable state. This is a contradiction as we should have had $\Tr_{23} (\sigma) = \rho_W$, which is an entangled state.
\end{exm}

\section{Conclusions} \label{sec:conclusions}
We have introduced the general definition of compatibility of channels in general probabilistic theory through the idea of conditional channels. We have also shown that a naive idea for a compatibility test leads to a simple and straightforward formulation of steering and Bell nonlocality. These formulations of steering and Bell nonlocality are overall new even when we consider only measurements instead of channels. Throughout the paper we have shown that all of our definitions and result are in correspondence with the known result for measurements and we have also provided several examples and results about the introduced concepts in quantum theory.

The paper has opened several new questions and areas of research. For example, a possible area of research would be to look at the structure of conditional states and conditional channels and to try to connect them to Bayesian theory.

Concerning the compatibility of channels, one may formulate different notions of degree of (in)compatibility or of robustness of compatibility in general probabilistic theory and look at their properties, in a similar way as it was already done in quantum theory \cite{Haapasalo-robustness}. For quantum channels one may wonder which types of channels are compatible. This would generalize the no broadcasting theorem \cite{BarnumHowardBarretLeifer-noBroadcast, BarnumBarretLeiferWilce-noBroadcast} which states that two unitary channels can not be compatible.

One may also consider our formulations of steering and Bell nonlocality as a case of the problem of finding a multipartite state with given marginals. Such problems were studied in recent years \cite{Klyachko-marginals, WyderkaHuberGuhne-marginals}, but not in the form that would be applicable to the problems of steering and Bell nonlocality as incompatibility tests. This opens questions whether one may characterize the structure of the cone $Q_{CD}$ and of other cones of interest in quantum theory. From a geometrical viewpoint this question is closely tied to the question of existence of other Bell inequalities for channels than the one we presented. Existence and exact form of the generalized Bell inequalities is also a very interesting possible area of research.

We may also consider the use of steering and Bell nonlocality of channels in the context of quantum information theory and quantum communication. Both steering and Bell nonlocality of measurement were used to formulate new quantum protocols and it is of great interest whether exploiting the steering and Bell nonlocality of channels may lead to even better or more useful applications.

One may also try and clarify the lack of connection between steering and Bell nonlocality of channels. As we have showed in example \ref{exm:quantnonlocal-nonlocalnotsteer}, even if two channels can not steer a state, when applied to both parts of the state the resulting biconditional bipartite state may be Bell nonlocal. This may even have interesting applications in quantum theory of information as so far steering has been considered to lead to one-side device-independent protocols that were seen as a middle step between the original protocol and device-independent protocol.

It may also be interesting to consider the resource theories of channel incompatibility, of steering by channels and of Bell nonlocality of channels. Several similar resource theories were already constructed, see \cite{HorodeckiOppenheim-resource} for a review.

\begin{acknowledgments}
The author is thankful to Anna Jen\v{c}ov\'{a}, Michal Sedl\'{a}k, M\'{a}rio Ziman, Daniel Reitzner and Tom Bullock for interesting conversations on the topic of compatibility. This research was supported by grant VEGA 2/0069/16 and by the grant of the Slovak Research and Development Agency under contract APVV-16-0073. The author acknowledges that this research was done during a PhD study at Faculty of Mathematics, Physics and Informatics of the Comenius University in Bratislava.
\end{acknowledgments}

\bibliography{citations}

\end{document}